\pdfoutput=1
\documentclass[11pt]{article}

\usepackage{indentfirst}
\usepackage{amsmath}
\usepackage{amsthm}
\usepackage{color}
\usepackage{eufrak}
\usepackage{verbatim}
\usepackage[makeroom]{cancel}
\usepackage{rotating}
\usepackage{float}
\usepackage[margin=1in]{geometry}
\usepackage{amssymb}
\usepackage{amsthm}
\usepackage{mathrsfs}
\usepackage{epsfig}
\usepackage{graphicx}
\usepackage{indentfirst}
\usepackage{framed}
\usepackage[numbers]{natbib}
\usepackage{color}
\usepackage{hyperref}
\usepackage{graphicx}
\usepackage{hyperref}
\usepackage{float}
\usepackage{bm}
\usepackage{cleveref}
\usepackage{enumitem}

\usepackage{tikz}
\usepackage{multicol}
\usetikzlibrary{positioning,shapes,shadows,arrows}

\usepackage[labelfont=bf]{caption}

\usepackage{algorithm}
\usepackage[noend]{algpseudocode}

\definecolor{corlinks}{RGB}{100,0,100}
\definecolor{cormenu}{RGB}{100,0,100}
\definecolor{corurl}{RGB}{100,0,100}

\hypersetup{
	colorlinks=true,
	urlcolor=corlinks,
	linkcolor=corlinks,
	menucolor=cormenu,
	citecolor=corlinks,
	pdfborder= 0 0 0
}

\newtheorem{theorem}{Theorem}
\newtheorem{question}[theorem]{Question}

\newtheorem{lemma}[theorem]{Lemma}
\newtheorem{corollary}[theorem]{Corollary}
\newtheorem{definition}[theorem]{Definition}
\newtheorem{proposition}[theorem]{Proposition}

\newtheorem{fact}[theorem]{Fact}
\newtheorem{claim}[theorem]{Claim}

\newtheorem{hypothesis}[theorem]{Hypothesis}

\DeclareMathOperator{\poly}{poly}

\DeclareMathOperator*\Prob{{\bf Pr}}
\newcommand{\bool}{\left\{0,1\right\}}
\newcommand{\Kt}{\mathsf{Kt}}
\newcommand{\rKt}{\mathsf{rKt}}

\def\colorful{1}

\ifnum\colorful=1

\fi
\ifnum\colorful=0

\fi

\newcommand{\eqdef}{\stackrel{\rm def}{=}}

\newcommand{\PSPACE}{\mathsf{PSPACE}}
\newcommand{\DTIME}{\mathsf{DTIME}}
\newcommand{\DSPACE}{\mathsf{DSPACE}}
\newcommand{\BPTIME}{\mathsf{BPTIME}}
\newcommand{\BPP}{\mathsf{BPP}}
\newcommand{\BPE}{\mathsf{BPE}}

\newcommand{\randkt}{\mathsf{rKt}}

\newcommand{\rsr}{$\mathsf{rsr}$}
\newcommand{\dsr}{$\mathsf{dsr}$}

\usepackage{titlesec}

\titleformat{\paragraph}
{\normalfont\normalsize\bfseries}{\theparagraph}{1em}{}
\titlespacing*{\paragraph}
{0pt}{3.25ex plus 1ex minus .2ex}{1.5ex plus .2ex}

\begin{document}
	
	\newgeometry{margin=0.8in}
	
	\title{Pseudodeterministic Algorithms and the Structure of Probabilistic Time\vspace{0.4cm}}

	\author{
		Zhenjian Lu\footnote{Email: \texttt{zhen.j.lu@warwick.ac.uk}}\vspace{0.2cm}\\
		{\small University of Warwick} 
		\and 
		Igor C. Oliveira\footnote{Email: \texttt{igor.oliveira@warwick.ac.uk}}\vspace{0.2cm}\\{\small University of Warwick\vspace{0.3cm}}
		\and
		Rahul Santhanam\footnote{Email: \texttt{rahul.santhanam@cs.ox.ac.uk}}\vspace{0.2cm}\\{\small ~\,University of Oxford}
		\vspace{0.4cm}
	}

	
	\maketitle
	
	\vspace{-0.7cm}
	
	\renewenvironment{abstract}
 {\small
  \begin{center}
  \bfseries \abstractname\vspace{-.5em}\vspace{0pt}
  \end{center}
  \list{}{%
    \setlength{\leftmargin}{14mm}
    \setlength{\rightmargin}{\leftmargin}%
  }%
  \item\relax}
 {\endlist}
	
	\begin{abstract}
		   We connect the study of pseudodeterministic algorithms to two major open problems about the structural complexity of $\mathsf{BPTIME}$: proving \emph{hierarchy theorems} and showing the existence of \emph{complete problems}. Our main contributions can be summarised as follows.\\

		\noindent  \textbf{A new pseudorandom generator and its consequences.}~We build on techniques developed to prove hierarchy theorems for probabilistic time with advice (Fortnow and Santhanam \citep{DBLP:conf/focs/FortnowS04}) to construct the first \emph{unconditional} pseudorandom generator of polynomial stretch computable in pseudodeterministic polynomial time (with one bit of advice) that is secure infinitely often against polynomial-time computations.  As an application of this construction, we obtain new results about the complexity of generating and representing prime numbers. For instance, we show unconditionally for each $\varepsilon > 0$ that infinitely many primes $p_n$ have a {\it succinct representation} in the following sense: there is a fixed probabilistic \emph{polynomial time} algorithm that generates $p_n$ with high probability from its succinct representation of size $O(|p_n|^{\varepsilon})$. This offers an exponential improvement over the running time of previous results, and shows that infinitely many primes have \emph{succinct} and \emph{efficient} representations.\\ 
		
		\noindent \textbf{Structural results for probabilistic time from pseudodeterministic algorithms.} Oliveira and Santhanam \citep{DBLP:conf/stoc/OliveiraS17} established unconditionally that there is a pseudodeterministic algorithm for the Circuit Acceptance Probability Problem ($\mathsf{CAPP}$) that runs in sub-exponential time and is correct with high probability over any samplable distribution on circuits on infinitely many input lengths. We show that improving  this running time or obtaining a result that holds for every large input length would imply new time hierarchy theorems for probabilistic time. In addition, we prove that a worst-case polynomial-time pseudodeterministic algorithm for $\mathsf{CAPP}$ would imply that $\mathsf{BPP}$ has complete problems.\\
		
		\noindent  \textbf{Equivalence between pseudodeterministic constructions and hierarchies.}~We establish an equivalence between a certain explicit pseudodeterministic construction problem and the existence of strong hierarchy theorems for probabilistic time. More precisely, we show that pseudodeterministically constructing in exponential time strings of large $\mathsf{rKt}$ complexity (Oliveira \citep{DBLP:conf/icalp/Oliveira19}) is possible if and only if for every constructive function $T(n) \leq \exp(o(\exp(n)))$ we have $\mathsf{BPTIME}[\mathsf{poly}(T)] \nsubseteq \mathsf{i.o.}\mathsf{BPTIME}[T]/\log T$. \\
		
		   More generally, these results suggest new approaches for designing pseudodeterministic algorithms for search problems and for unveiling the structure of probabilistic time.

	\end{abstract}
	
	\restoregeometry
	
	\newpage
	
	\tableofcontents

	\newpage
	
	\section{Introduction}
	
	A pseudodeterministic algorithm for a search problem $\cal S$ is a probabilistic algorithm that with high probability outputs a {\it fixed} solution to $\cal S$ on any given input. The notion of pseudodeterminism was pioneered by Gat and Goldwasser \cite{DBLP:journals/eccc/GatG11}, motivated by applications in cryptography and distributed computing. Pseudodeterminism has been the topic of much recent work and has been studied in a variety of settings, including query complexity, property testing,  parallel computation, learning algorithms,  space-bounded computation, streaming algorithms and interactive proof systems \cite{DBLP:journals/eccc/GatG11, DBLP:conf/innovations/GoldreichGR13, DBLP:journals/eccc/GoldwasserG15, DBLP:journals/eccc/Grossman15,  DBLP:conf/stoc/OliveiraS17, DBLP:journals/corr/Holden17, DBLP:conf/approx/OliveiraS18, DBLP:conf/innovations/GoldwasserGH18, DBLP:conf/mfcs/DixonPV18, DBLP:conf/soda/GrossmanL19, DBLP:journals/eccc/GoemansGH19, DBLP:conf/icalp/Oliveira19, DBLP:journals/eccc/Goldreich19,  DBLP:conf/innovations/GoldwasserGMW20, paperITCS21}.
	
	A fundamental question about pseudodeterministic algorithms posed in \cite{DBLP:journals/eccc/GatG11} is whether there is a polynomial-time pseudodeterministic algorithm for generating prime numbers of a given length. Note that there is a trivial {\it probabilistic} algorithm that generates a random number with $n$ bits and checks it for primality; however, this algorithm is far from being pseudodeterministic.
	
	The question of efficient generation of primes has attracted broad interest, including the Polymath 4 project \cite{MR2869058} devoted to this topic. Despite this, known unconditional results are still fairly weak: the most efficient deterministic algorithm \cite{DBLP:journals/jal/LagariasO87} to generate $n$-bit primes runs in time $\Omega(2^{n/2})$.
	In \cite{DBLP:conf/stoc/OliveiraS17}, some progress was made on the question of \cite{DBLP:journals/eccc/GatG11} about generating primes. They give a pseudodeterministic algorithm running in time $2^{n^{o(1)}}$ on input of length $1^n$ that generates a fixed prime $p_n$ with high probability for infinitely many $n$. While this algorithm is a significant improvement on brute force search, it is unsatisfactory in a couple of different respects: it runs in sub-exponential time rather than polynomial time, and it is only guaranteed to be correct for infinitely many $n$.
	
	Somewhat surprisingly, the algorithm of \cite{DBLP:conf/stoc/OliveiraS17} uses very little information about primes -- just that they are plentiful (by the Prime Number Theorem), and that there is a polynomial-time algorithm for Primality \cite{Agrawal02primesis}. Indeed, \cite{DBLP:conf/stoc/OliveiraS17} show a far more general result giving a pseudodeterministic algorithm solving the search version of the Circuit Acceptance Probability Problem ($\mathsf{CAPP}$), from which the prime generation result follows easily. This more general result has, of course, the same caveats as in the result for primes: the running time is sub-exponential, and the success of the pseudodeterministic algorithm is only guaranteed for infinitely many input lengths.
	
	Strengthening this general result to algorithms that run in polynomial time and work for almost all input lengths would solve the main open question of \cite{DBLP:journals/eccc/GatG11}, hence it is natural to wonder if this is possible. In this paper, we show that progress on this question is tightly connected to longstanding open problems about the structure of probabilistic time, namely the question of whether $\mathsf{BPP}$ has {\it complete problems} and the question of whether there is a {\it hierarchy theorem} for $\BPP$. We show that these connections go in {\it both} directions: we exploit previous work on hierarchies for probabilistic time to show new results on pseudodeterministic generation of primes, and we show that any improvements in the general result of \cite{DBLP:conf/stoc/OliveiraS17} would yield progress on hierarchies and complete problems for $\mathsf{BPTIME}$.
	
	We briefly review what is known about the structure of probabilistic time. Recall that $\mathsf{BPP}$ is the class of decision problems solvable in polynomial time by a probabilistic machine that has bounded error on every input. $\mathsf{BPP}$ is a {\it semantic} class rather than a {\it syntactic} one, meaning that there is no canonical enumeration of machines defining those and only those languages in the class. The reason is that the acceptance and rejection criteria for a probabilistic machine $M$ on an input are not exhaustive -- it could be that a machine $M$ satisfies its bounded-error promise on some inputs but not others, in which case it does not define a language in $\mathsf{BPP}$. Indeed, it is not hard to show that it is {\it undecidable} whether a given probabilistic machine $M$ satisfies its bounded-error promise on every input. In contrast, for syntactic classes such as $\mathsf{P}$, $\mathsf{NP}$ and $\mathsf{PSPACE}$, the acceptance and rejection criteria are indeed exhaustive -- a given deterministic or non-deterministic machine accepts or rejects on any given input. Syntactic classes have canonical complete problems which are based on canonical enumerations of machines defining the class, but semantic ones do not. Under strong derandomization assumptions, $\mathsf{BPP} = \mathsf{P}$ \cite{DBLP:conf/stoc/ImpagliazzoW97}, and hence $\mathsf{BPP}$ has complete problems because $\mathsf{P}$ does, but we know nothing at all about the existence of complete problems {\it unconditionally}. In fact, we do not even know if complete problems exist for the class of problems solvable on average in probabilistic polynomial time, or the class of problems solvable in probabilistic polynomial time with small advice. 
	
	The semantic nature of the class $\mathsf{BPP}$ is also relevant to the existence of {\it hierarchy theorems} for the class. A hierarchy theorem is a result showing unconditionally that more resources allow us to solve more problems. Some of the earliest results in complexity theory \cite{DBLP:journals/jacm/HennieS66,DBLP:conf/focs/StearnsHL65}
	were hierarchy theorems for deterministic time and space. Almost optimal hierarchy theorems are known for {\it every} syntactic class \cite{DBLP:journals/jcss/Cook73, DBLP:journals/jacm/SeiferasFM78, Zak83} with resource bounds up to exponential, using diagonalization arguments. However, these diagonalization arguments presuppose that there is an efficient canonical enumeration of machines in the class, and hence do not work for semantic classes.
	
	By using padding arguments and exploiting hierarchies for deterministic time, it is known that $\mathsf{BPP}$ is strictly contained in $\mathsf{BPSUBEXP}$ \cite{DBLP:journals/iandc/KarpinskiV87}, but it is still open even whether $\mathsf{BPTIME}(n)$ is strictly contained in $\mathsf{BPTIME}(T(n))$, for any function $T$ that remains sub-exponential even when composed with itself a constant number of times. The situation is slightly better when it comes to hierarchies for variants of $\mathsf{BPP}$: in a line of works \cite{DBLP:conf/random/Barak02, DBLP:conf/focs/FortnowS04}, hierarchies were shown for $\mathsf{BPP}/1$ (the class of problems solvable in probabilistic polynomial time with 1 bit of advice) and for $\mathsf{Heur}$-$\mathsf{BPP}$ (the class of problems solvable on average in probabilistic polynomial time). Despite much effort, it remains wide open to show a hierarchy for $\BPP$.
	
	Note that these questions about the structure of probabilistic time are about {\it separations} (in the case of hierarchies) and about {\it hardness} (in the case of complete problems), while the question of pseudodeterministic constructions is an {\it algorithmic} question. Connections between algorithms and lower bounds have already been very fruitful in complexity theory, e.g., in the theory of pseudorandomness or in Williams' algorithmic method for complexity lower bounds \cite{DBLP:journals/siamcomp/Williams13}. We provide yet another instance of this phenomenon.
	
	We now describe our results in more detail.

	\subsection{Results}\label{sec:results}
	
	Our first results show how to obtain new pseudodeterministic constructions by building on techniques employed to establish hierarchy theorems.\\
	
	\noindent \textbf{A new pseudorandom generator and improved bounds for primes.} Our main technical result is an unconditional construction of a pseudorandom generator (PRG) with seed length $n^{\varepsilon}$ that is secure infinitely often against uniform adversaries. The generator is computable in probabilistic polynomial time with one bit of advice. Note that while a {\it random} function from $n^{\varepsilon}$ bits to $n$ bits is a PRG with high probability, it is non-trivial to compute such a generator {\it efficiently} and {\it pseudodeterministically}.
	
		\begin{theorem}[A pseudodeterministic polynomial-time computable PRG with $1$ bit of advice] \label{t:intro_PRG}~\\
		For every $\varepsilon > 0$ and $c,d \geq 1$, there exists a generator $G = \{G_n\}_{n \geq 1}$ with $G_n \colon \{0,1\}^{n^\varepsilon} \to \{0,1\}^n$ for which the following holds:
		\begin{itemize}[leftmargin=*]
		    \item[] \emph{Efficiency:} There is a probabilistic polynomial-time algorithm $A$ that given $n$,  $x \in \{0,1\}^{n^\varepsilon}$, and an advice bit $\alpha(n) \in \{0,1\}$ that is independent of $x$, outputs $G_n(x)$ with probability $\geq 2/3$.
		    \item[] \emph{Pseudorandomness:} For every  language $L \in \mathsf{DTIME}[n^c]$, there exist infinitely many input lengths $n$ such that 
		    $$
		    \Big | \Pr_{y \sim \mathcal{U}_{n}}[L(y) = 1] - \Pr_{x \sim \mathcal{U}_{n^\varepsilon}}[L(G_n(x))= 1]   \Big | \leq  \frac{1}{n^d}.
		    $$
		\end{itemize}
	\end{theorem}
	
	\vspace{0.2cm}
	
	In contrast, the pseudodeterministic generator from \citep{DBLP:conf/stoc/OliveiraS17} has a running time that is \emph{exponential} in the seed length $n^{\varepsilon}$. We remark that the security of $G$ also holds against randomized algorithms (modelled as a samplable distribution of circuits); see Section \ref{s:PRG_construction} for details.
	  
	  As a corollary of Theorem \ref{t:intro_PRG}, we obtain a new result about pseudodeterministic \emph{polynomial-time} construction of primes. Recall that, to solve the main open question of \cite{DBLP:journals/eccc/GatG11}, we need to show that there is a pseudodeterministic polynomial-time algorithm $A$ such that $A(1^n)$ is a prime for all $n$. We make progress on this by giving a pseudodeterministic algorithm that succeeds when given a {\it succinct} representation of $p_n$, rather than just $n$ in unary. Thus, it is possible to compress infinitely many primes such that these primes can be recovered efficiently and pseudodeterministically from the compressed representation. To the best of our knowledge, nothing non-trivial was known about constructions of primes in the \emph{polynomial time} regime.\\

		\noindent \textbf{Corollary} (Existence of infinitely many primes with short and efficient descriptions).\\
	\emph{For every $\varepsilon > 0$, there is a probabilistic polynomial-time algorithm $A$ and a sequence $\{p_m\}_{m \geq 1}$ of increasing primes $p_m$ such that there exist a sequence $\{a_m\}_{m \geq 1}$ of strings, with $|a_m| = |p_m|^{\varepsilon}$, for which $A(a_m) = p_m$ with high probability for each $m$.}\\
	
	As another consequence, we get that there is a probabilistic polynomial-time algorithm that on input $1^n$ outputs a fixed prime of length $n$ with probability $2^{-n^{\varepsilon}}$ for infinitely many $n$. Indeed, we just simulate the algorithm $A$ in the corollary above and guess the input $a_m$ at random given $m$ in unary. To the best of our knowledge, prior to our work, there was no probabilistic polynomial-time algorithm that generated a fixed $n$-bit prime with success probability $2^{-o(n)}$.
	
	Theorem \ref{t:intro_PRG} also has implications for the study of Kolmogorov complexity. Indeed, the results mentioned above admit natural formulations in this language. We define a new notion of Kolmogorov time-bounded randomized complexity $\mathsf{rK^{\poly}}$, which measures the smallest size of a program from which a given string $x$ can be generated with high probability in polynomial time. Theorem \ref{t:intro_PRG} implies that for every $\varepsilon > 0$, every dense set in $\mathsf{P}$ has strings of length $n$ with $\mathsf{rK^{\poly}}$ complexity at most $n^{\varepsilon}$, for infinitely many $n$. In addition, it is possible to employ Theorem \ref{t:intro_PRG} to prove unconditional complexity lower bounds for the problem of estimating the $\mathsf{rK^{\poly}}$ complexity of an input string. We refer the interested reader to Sections \ref{s:prelim_kolmogorov} and \ref{s:applications} for more details.\\
	
	Next, we show connections in the \emph{reverse} direction between pseudodeterministic algorithms and structural results for probabilistic time, i.e., that better hierarchy theorems and structural results for probabilistic time can be obtained from better pseudodeterministic algorithms. We use the term ``pseudo-derandomisation'' (PD) to refer to the simulation of a randomized algorithm for a search problem by a pseudodeterministic one.\\
		
	\noindent \textbf{Mildly better pseudo-derandomisations yield new structural results for $\mathsf{BPTIME}$.} For a positive integer $d$, we define $\mathsf{CAPP}_{n,n^d}$ to be the search problem where given as input $x = (1^n,C)$, where $|C|=n^d$ and $C$ is interpreted as a Boolean circuit on at most $n^d$ input variables and of size at most $n^d$, we must output a number $\mu \in [0,1]$ such that
		\[
		\left|\Prob_{y\in\bool^{n^{d}}}[C(y)=1]-\mu \, \right|\;\leq\; 1/10.
		\]
	We recall the following \emph{unconditional} result established by Oliveira and Santhanam \citep{DBLP:conf/stoc/OliveiraS17}. 
	
	\vspace{0.2cm}
	
	\noindent ($\star$)~\emph{Infinitely-often average-case sub-exponential time pseudo-derandomisation of $\mathsf{CAPP}$}: 
	
		\vspace{-0.2cm}
	
	\begin{itemize}[leftmargin=*]
	    \item[] For any $\varepsilon>0$ and $c, d \geq 1$, there is a pseudodeterministic algorithm for $\mathsf{CAPP}_{n,n^{d}}$ that runs in time $T(n) = 2^{n^{\varepsilon}}$, and for any polynomial-time samplable ensemble of distributions $\mathcal{D}_{n, n^d}$ of circuits of size $n^d$, succeeds with probability $1-1/n^c$ over $\mathcal{D}_{n, n^d}$ for infinitely many values of $n$. 
	\end{itemize}
	
	We stress that when referring to a \emph{pseudo-deterministic} algorithm $A$ that succeeds infinitely often and on average, we still require that on \emph{every} input string $x$, $A(x)$ produces a canonical output with high probability. The aforementioned result satisfies this property (see Appendix \ref{s:appendix_OS17} for more details).
	
	The statement in ($\star$) has several caveats: the running time is exponential, the simulation only succeeds infinitely often, and the pseudo-deterministic algorithm might fail to produce a correct answer on some inputs (i.e. the canonical output might not be an accurate estimate of the acceptance probability of the input circuit). 
	
	The next statement shows that addressing \emph{any} of these caveats would imply new structural results for probabilistic time (even if the pseudodeterministic algorithm  depends on the samplable distribution of inputs).

	\begin{theorem}[Structural results for $\mathsf{BPTIME}$ from  better pseudo-derandomisations, Informal]\label{t:pseudo-to-structural}~\\
	Let $d \geq 1$, and suppose that for each polynomial-time samplable distribution of input circuits, there is a pseudodeterministic algorithm $A$ that solves $\mathsf{CAPP}_{n, n^{d+1}}$ \underline{infinitely often} \underline{on average} in time $T(n)$. Then,
	\begin{itemize}
	    \item[\emph{(}i\emph{)}] There is a language $L \in \mathsf{BPTIME}[T(n)] \setminus \mathsf{BPTIME}[n^d]$.
	    \item[\emph{(}ii\emph{)}] Moreover, if $A$ succeeds \underline{almost everywhere}, then there is $L \in \mathsf{BPTIME}[T(n)] \setminus \mathsf{i.o.}\mathsf{BPTIME}[n^d]$.
	    \item[\emph{(}iii\emph{)}] Finally, if $A$ is correct \underline{almost everywhere} and in the \underline{worst case}, there exist $\mathsf{BPTIME}$-hard problems in $\mathsf{BPTIME}(T(n))$.
	\end{itemize}
		\end{theorem}

		Item (\emph{i}) shows that improving the \emph{running time} $T(n)$ of the algorithm in $(\star)$ would lead to a new hierarchy theorem with tighter time bounds. On the other hand, from Item (\emph{ii}) we get that removing the \emph{infinitely often} condition from $(\star)$ would prove the first hierarchy result against $\mathsf{i.o.}\mathsf{BPTIME}[n^d]$, i.e., when  the language is hard on every large enough input length. Finally, Item (\emph{iii}) shows how to obtain complete problems from \emph{worst-case} pseudo-derandomisations.  We note that weaker consequences can also be obtained by relaxing the assumptions from Item (\emph{iii}). Indeed, a \emph{new} average-case infinitely often completeness result follows from $(\star)$. Since it is somewhat technical to formally state the result, we refer to the body of the paper for details (Corollary \ref{c:new_completeness}).
		
	Consequences of a similar nature also follow from weak pseudo-derandomisations of algorithms solving problems in $\mathsf{BPP}$-$\mathsf{search}$. Since the formulation of the  results for $\mathsf{BPP}$-$\mathsf{search}$ might be of independent interest, we include them in Appendix \ref{s:BPP_search_appendix}.\\

		Finally, we show a setting where hierarchy theorems and pseudo-derandomisations turn out to be \emph{equivalent}.\\
		
	\noindent \textbf{An equivalence between pseudo-derandomisation and hierarchies.} Our last result shows the existence of an explicit construction problem that is in a certain sense \emph{universal} for probabilistic time hierarchies. 
	
	In order to state the result, we recall the following fundamental notion from Kolmogorov complexity introduced by Levin \citep{DBLP:journals/iandc/Levin84}. For a string $x$, $\mathsf{Kt}(x)$ is defined as the minimum value $|M| + |a| + \log (t)$ over all tuples $(M,a,t)$, where $a$ is an arbitrary string, $M$ is a deterministic machine that prints $x$ in $t$ steps when it is given  $a$ as input, and $|M|$ is the length of its representation as a binary string (according to a fixed universal machine). 
	
	It is not hard to show that given $n$ we can construct a string $x$ such that $\mathsf{Kt}(x) \geq n$ in deterministic time $2^{O(n)}$.
	
	Oliveira \citep{DBLP:conf/icalp/Oliveira19} introduced a natural randomized analogue of Levin's definition, denoted $\mathsf{rKt}(x)$. The only difference is that now the minimization takes place over all tuples $(M,a,t)$ where $M$ is a \emph{randomized} machine that outputs $x$ with probability at least $2/3$ when it computes for $t$ steps on input $a$. We refer to Section \ref{s:prelim} for a precise definition.

	Can we construct in probabilistic exponential time a (fixed) string of large $\mathsf{rKt}$ complexity? 

	\begin{theorem}[An equivalence between pseudo-derandomisation and probabilistic time hierarchies] \label{t:equivalences} The following statements are equivalent:
		\begin{itemize}
			\item[\emph{(1)}] \emph{Pseudodeterministic construction of strings of large $\rKt$ complexity:} 
			There is a constant $\varepsilon > 0$ and a randomised algorithm $A$ that, given $m$, runs in time $2^{O(m)}$ and outputs with probability at least $2/3$ a fixed $m$-bit string $w_m$ such that $\rKt(w_m) \geq \varepsilon m$.
			\item[\emph{(2)}] \emph{Strong time hierarchy theorem for probabilistic computation:} There are constants $k \geq 1$ and $\lambda > 0$ for which the following holds. For any constructive function $n \leq t(n) \leq 2^{\lambda \cdot 2^n}$, there is a language $L \in \mathsf{BPTIME}[(t(n)^k]$ such that $L \notin \mathsf{i.o.BPTIME}[t(n)]/\log t(n)$. 
					\end{itemize}
	\end{theorem}

Consequences from weaker pseudodeterministic constructions of strings of non-trivial $\rKt$ complexity are explored in Section \ref{s:hier_from_expl_constr}.

We conjecture that the equivalence from Theorem \ref{t:equivalences} extends to capture the pseudo-derandomisation of unary problems in $\mathsf{BPE}$-$\mathsf{search}$, and we elaborate on this in Section \ref{sec:unary_BPE}.\\

We summarise several connections established in our paper in Appendix \ref{s:summary_appendix}. \\
\\
{\bf \noindent Relationship to the independent work of \cite{DPV21b}}: Peter Dixon, A. Pavan and N. V. Vinodchandran have very recently brought to our attention that they have independent and concurrent unpublished work \cite{DPV21b} that overlaps with this work, and have shared a draft with us. We briefly discuss the relationship between their work and ours.

Like our work, \cite{DPV21b} show close connections between pseudodeterministic algorithms for $\mathsf{CAPP}$ and structural results on $\mathsf{BPP}$. In particular, as in our Theorem \ref{t:pseudo-to-structural}, they show that polynomial-time pseudodeterministic algorithms for $\mathsf{CAPP}$ imply a hierarchy for $\mathsf{BPP}$.

There are also differences between the two works. Directions such as unconditional constructions of pseudodeterministic PRGs and equivalences between hierarchy theorems and pseudo-deterministic constructions of strings of high $\rKt$ complexity, corresponding to Theorem \ref{t:intro_PRG} and Theorem \ref{t:equivalences} in this work, are not explored in \cite{DPV21b}. There are also results in \cite{DPV21b}, such as fixed-polynomial circuit lower bounds for $\mathsf{MA}$ and conversion of multi-pseudodeterministic algorithms to pseudodeterministic algorithms, both shown under the assumption that $\mathsf{CAPP}$ has polynomial-time pseudodeterministic algorithms, which do not have counterparts in our work.

	\subsection{Techniques}
	
	In this section, we provide an overview of our main ideas and techniques. We start with a discussion of our most technically demanding result showing how to obtain new pseudodeterministic algorithms from existing probabilistic time hierarchies with advice.\\

	\noindent \textbf{Theorem \ref{t:intro_PRG}:~Pseudodeterministic algorithms from hierarchies.}~In trying to derive an implication from hierarchy theorems for probabilistic time to pseudodeterministic algorithms, our starting point is the observation in \cite{DBLP:conf/stoc/OliveiraS17} that exponential circuit lower bounds for $\mathsf{BPE}$ can be used to get a pseudodeterministic poly-time computable PRG with seed length $O(\log(n))$, just by plugging in a hard function in $\mathsf{BPE}$ into the Impagliazzo-Wigderson generator \cite{DBLP:conf/stoc/ImpagliazzoW97}. Such a PRG is secure even against non-uniform adversaries; if we only need security against uniform adversaries, intuitively it should suffice to start with a hard function in $\mathsf{BPE}$ against sub-exponential time probabilistic time, i.e., a hierarchy theorem. If this approach worked, we would actually be able to get pseudodeterministic poly-time generation of primes, just by listing the outputs of the PRG in lexicographic order and outputting the first one that passes the Primality test.
	
	There are a couple of problems with this. First, since we do not know a hierarchy theorem for $\mathsf{BPE}$, we cannot hope to get an unconditional result this way. Second, even if our goal is merely to get a connection between hierarchy theorems and pseudodeterministic algorithms, known techniques for arguing security against uniform adversaries \cite{DBLP:journals/jcss/ImpagliazzoW01, DBLP:journals/cc/TrevisanV07} do not work when starting with a function in $\mathsf{BPE}$. Rather, they require the hard function to be downward self-reducible, and downward self-reducibility implies that the hard function is in $\mathsf{PSPACE}$.
	
	To get around these problems, we start with a hard problem in $\mathsf{BPP}/1$ rather than in $\mathsf{BPE}$. We lose something by doing this -- now we can only hope for PRGs with seed length $n^{\varepsilon}$ rather than $O(\log(n))$. But we also gain something, as we know from the $\mathsf{BPP}$ hierarchy theorem with advice \cite{DBLP:conf/random/Barak02, DBLP:conf/focs/FortnowS04} that hard languages unconditionally exist: for every $k$ there is a language $L_k$ in $\mathsf{BPP}/1$ that is not in $\mathsf{BPTIME}(n^k)/1$. Now we can try to plug in the language $L_k$ into the amplified Nisan-Wigderson generator as used in the uniform hardness-to-randomness reduction of
	\cite{DBLP:journals/jcss/ImpagliazzoW01, DBLP:journals/cc/TrevisanV07}. However, this reduction is inherently non-black-box and requires the initial language $L_k$ to be downward self-reducible and random self-reducible.
	
	Starting with an arbitrary hard language in $\mathsf{BPP}/1$, we do not know how to transform it into one that satisfies the properties required by the reduction in \cite{DBLP:journals/jcss/ImpagliazzoW01, DBLP:journals/cc/TrevisanV07}, while maintaining hardness. We are free though to design a hard language $L_k$ ourselves, rather than starting with an arbitrary one. Our idea is to exploit the \emph{structure} of the hard language $L_k$ in the hierarchy theorem of \cite{DBLP:conf/focs/FortnowS04}.
	
	What is promising is that the proof of the hierarchy theorem in \cite{DBLP:conf/focs/FortnowS04} starts with a certain structured $\mathsf{PSPACE}$-complete problem $L_{hard}$ with special properties constructed in \cite{DBLP:journals/cc/TrevisanV07}. The hard language $L_k$ in the hierarchy theorem of \cite{DBLP:conf/focs/FortnowS04} is a \emph{padded version} of $L_{hard}$. By modifying $L_k$ slightly so that the padding does not lose information, we can hope to show that the language $L_k$ {\it inherits} random self-reducibility (\rsr) and downward self-reducibility (\dsr) from the language $L_{hard}$. This would enable us to plug the modified version of $L_k$ into the hardness-to-randomness reduction of \cite{DBLP:journals/cc/TrevisanV07} and thus show security against uniform adversaries.
	
	Unfortunately, it is not quite true that the language $L_k$ inherits $\mathsf{rsr}$ and $\mathsf{dsr}$ from $L_{hard}$. Indeed, the amount of padding required to transform $L_{hard}$ into $L_k$ is not efficiently computable in general, and this is the reason one bit of advice is required to decide $L_k$ with a $\mathsf{BPP}$ algorithm. As a consequence, $L_k$ is only $\mathsf{rsr}$ with one bit of advice, and similarly $\mathsf{dsr}$ given the right bit of advice.
	
	This turns out to be an issue when using the learning procedure of \cite{DBLP:journals/cc/TrevisanV07}, which builds up a circuit for the hard function at length $n$ from failure of the PRG at length $n$ by inductively building circuits at length $i$ for $i < n$ and then using $\mathsf{rsr}$ and $\mathsf{dsr}$ to complete the inductive step. If one bit of advice is required at each input length, then $n$ bits of advice are required in all, and this kills the argument -- we do not known that $L_k$ is still hard for $\BPTIME(n^k)$ with $n$ bits of advice.
	
	We circumvent this using a modified learning strategy using the structure of the language $L_k$. The crucial observation is that the bit of advice in the $\mathsf{BPP}/1$ algorithm for $L_k$ is only used to tell if the input length is ``good'' in the sense of the padding being long enough. We show that for each good input length $n$, there is a sequence of smaller good input lengths such that the learning strategy can be implemented within these input lengths. Since these smaller input lengths inherit their goodness from the original length $n$, we do not need additional advice when using $\mathsf{rsr}$ and $\mathsf{dsr}$ at the smaller input lengths. This enables us to use the learning strategy to derive a $\mathsf{BPTIME}(n^k)/1$ algorithm for $L_k$, which is indeed a contradiction to the hardness of $L_k$.
	
	The above description omits many technical subtleties, but does convey the gist of the proof.\\
	\\
	\noindent \textbf{Theorem \ref{t:pseudo-to-structural}:~Hierarchies from pseudo-derandomisations.}~To show that weak (infinitely often and on average) pseudo-derandomisations of $\mathsf{CAPP}$ give hierarchy theorems for probabilistic time, we use \emph{diagonalization}. Suppose we want to diagonalize against randomized machines running in time $n^k$, while maintaining the promise that \emph{every} input string is either accepted or rejected with probability \emph{bounded away} from $1/2$. One way to proceed might be by first obtaining an \emph{estimate} of the acceptance probability of each input machine $M$ (say when running it on its code) via simulations of $M$, then flipping the output (i.e.~output $1$ if the estimate is less than $1/2$). One issue with this approach is that it is not clear how to implement this idea and put the ``diagonalized'' hard language in $\mathsf{BPTIME}$. The issue is that if the input machine $M$ for the simulation accepts certain strings with probability near $1/2$, we cannot guarantee to have a \emph{fixed} output bit with high probability (since the output depends on the estimate of the acceptance probability of $M$). 
	
	This is where \emph{pseudodeterminism} comes in helpful. If we can estimate the acceptance probability pseudodeterministically, which means we get a \emph{fixed} (though not necessarily correct) estimate with high probability, we can put the diagonalized language in $\mathsf{BPTIME}$. Crucially, $\mathsf{CAPP}$ is precisely the problem that allows one to estimate the acceptance probability of a randomized machine computing in bounded time, since we can obtain a circuit to describe the computation of $M$ on a given input as a function of its random string. However, we still need to address the fact that the pseudodeterministic simulation can make mistakes on some inputs, which could destroy the hardness of the diagonalized language.
	
	To cope with the issue that our pseudodeterministic algorithm $A$ for $\mathsf{CAPP}$ gives a correct answer only over a set $S \subseteq \mathbb{N}$ of input lengths (e.g.~in the infinitely often case $S$ is only guaranteed to be infinite), we use a careful padding technique to ensure that, for each fixed machine $M$, if the input length is large enough then $A$ attempts to diagonalize against $M$ over that input length. The only issue left is that, even on ``good'' input lengths (where goodness is determined by $S$), the pseudodeterministic algorithm succeeds only with high probability (say $\geq 1-1/n^2$) over the samplable distribution of circuits. This is handled by the observation that, thanks to our padding construction and the choice of an appropriate polynomial-time samplable distribution of input instances for $\mathsf{CAPP}$, the number of \emph{relevant} inputs (describing circuits obtained from machines) on each input length is small (say $\leq 2n$). This means that if $A$ is correct with high probability over the samplable distribution of interest employed in the diagonalization argument, it is also correct with high probability on \emph{each} relevant input string (i.e.~circuit). This idea can be formalised to show that the language produced through the diagonalisation process is indeed hard (infinitely often or almost everywhere, depending on $S$).\\

	\noindent \textbf{Theorem \ref{t:equivalences}:~Equivalence.}~For the equivalence between constructing strings of large $\rKt$ complexity and the existence of strong hierarchy theorems for probabilistic time, we proceed as follows. (For simplicity, we focus on the qualitative aspect of the proof.) We first observe that any string of large $\rKt$ complexity cannot be pseudodeterministically computed by randomized algorithms with small running time and with a small amount of advice -- this follows from the definition of $\rKt$. Therefore, if the truth table of a language \emph{contains} a string of large $\rKt$ complexity, the language cannot be computed in small $\mathsf{BPTIME}$ with a bounded amount of advice, since otherwise this string can be pseudodeterministically reconstructed from a probabilistic algorithm for the language and the correct advice. Using this idea, it is possible to employ a pseudodeterministic construction of strings of large $\rKt$ complexity to embed these strings in the definition of a (hard) language. This shows that a pseudodeterministic solution to the explicit construction problem for $\rKt$ yields a probabilistic time hierarchy with languages that are hard even against probabilistic algorithms with advice.
	
	For the other direction, suppose we have a fixed language $L$ in BPTIME that is hard against probabilistic algorithms of smaller running time, \emph{even with advice}. Then by viewing this language as a sequence of strings $\{w_n\}$ (obtained from the corresponding truth tables), we get that the probabilistic algorithm that decides $L$ can be transformed into a pseudodeterministic algorithm that generates $\{w_n\}$. We claim that this is a sequence of strings of large $\rKt$ complexity. Indeed, if not, then an optimal sequence of probabilistic machines $M_n$ that describe each $w_n$ (according to the definition of $\rKt$) can be given as \emph{advice} to a \emph{uniform} probabilistic algorithm that computes $L$ in bounded probabilistic time. This is a contradiction to the hardness of $L$.
	
	This completes the sketch of the equivalence between the two statements. Checking that the parameters obtained from a formalisation of the sketch given above are appropriate is not difficult.

	\section{Preliminaries} \label{s:prelim}

	\subsection{Basic definitions and notation}
	
	A function $t \colon \mathbb{N} \to \mathbb{N}$ is said to be time-constructible if there is a deterministic machine $M$ that on input $1^n$ halts within $O(t(n))$ steps and outputs $t(n)$. For simplicity, we might simply say that a function is constructible in this case. We say that $t$ is monotone if $t(a) \geq t(b)$ for $a \geq b$. 
	
	We write $|x|$ to denote the length of a string $x \in \{0,1\}^*$.
	
	The uniform distribution over $\{0,1\}^m$ is denoted by $\mathcal{U}_m$.
	
	For a function $\mu \colon \mathbb{N} \to [0,1]$, we say that a language $L \subseteq \{0,1\}^*$ is $\mu$-\emph{dense} if for every large enough $n$, we have $\Pr_{y \sim \{0,1\}^n}[y \in L] \geq \mu(n)$.
	
	We use $\mathsf{SIZE}[s]$ to refer to the class of languages that are computable by a sequence of circuits of size $s(n)$.
	
	We say that an ensemble $\{\mathcal{D}_n\}_{n \geq 1}$ of distributions is samplable in time $T(n)$ if there is a deterministic algorithm $A$ such that, for every $n$, the distribution induced by $A(1^n,\mathcal{U}_{T(n)})$ is $\mathcal{D}_n$ and $A(1^n,z)$ runs in time at most $T(n)$.
	
	\subsection{Probabilistic computations and search problems}
	
	We use $\mathsf{BPTIME}[t(n)]/a(n)$ to denote the set of languages computed in probabilistic time $O(t(n))$ using $a(n)$ bits of advice. Note that the acceptance probability of a machine with incorrect advice can be arbitrary. 
	
	In the definition below, we consider a binary relation $R \subseteq \{0,1\}^* \times \{0,1\}^*$ such that, for every $x \in \{0,1\}^*$, the set of solutions $R_x \eqdef \{y \mid (x,y) \in R\}$ is nonempty.

	\begin{definition}[$\mathsf{BPP}$-$\mathsf{search}$]\label{d:search-BPP}
		A binary relation $R$ is in $\mathsf{BPP}$-$\mathsf{search}$ if there exist both
		\begin{itemize}
			\item \emph{(Search algorithm)} a probabilistic polynomial-time algorithm $A$ such that for every input $x$, $A$ outputs $y \in R_x$  with probability at least $2/3$ over its internal randomness,
			\item \emph{(Verification algorithm)} and a probabilistic polynomial-time algorithm $B$ such that
			\begin{itemize}
				\item for every pair $(x,y)$, if $(x,y)\not\in R$ then $B$ rejects $(x,y)$ with probability at least $2/3$,
				\item and for every $x$,  with probability at least $1/2$ over the random choices of $A$ on input $x$, $B$ accepts $(x,A(x))$ with probability at least $2/3$.
			\end{itemize}
		\end{itemize}
		If this is the case, we say that the pair $(A,B)$ witnesses that $R \in \mathsf{BPP}$-$\mathsf{search}$.	
	\end{definition}
	
	Note that if $R \in \mathsf{BPP}$-$\mathsf{search}$ then using algorithms $A$ and $B$ from above we can efficiently find for a given $x$ a solution $y \in R_x$ and certify its validity with high probability. On the other hand, it is not necessarily the case that the relation $R$ can be efficiently decided, since the verification algorithm is not required to accept with high probability \emph{every} pair $(x,y) \in R$.\\

	\noindent \textbf{$\mathsf{BPE}$-$\mathsf{search}$ and $\mathsf{unary}$-$\mathsf{BPE}$-$\mathsf{search}$.}	Definition \ref{d:search-BPP} can be extended to algorithms $A$ and $B$ running in exponential time $2^{O(n)}$ as a function of $n = |x|$, which gives rise to the class of relations $\mathsf{BPE}$-$\mathsf{search}$. 
	
	We can also consider the class of \emph{unary} relations $R$, meaning that if $(x,y) \in R$ then $x = 1^n$ for some $n$, and for every $n$ there exists $y$ such that $(1^n,y) \in R$. The class $\mathsf{unary}$-$\mathsf{BPE}$-$\mathsf{search}$ is then defined in the natural way. More precisely, in Definition \ref{d:search-BPP} we restrict to $x$ of the form $1^n$, and allow exponential time algorithms $A$ and $B$ as in the case of $\mathsf{BPE}$-$\mathsf{search}$.\\
	
	\noindent \textbf{Pseudodeterministic algorithms for $\mathsf{BPP}$-$\mathsf{search}$ and $\mathsf{BPE}$-$\mathsf{search}$.} We say that a randomized algorithm $A$ pseudo-deterministically solves a search problem $R$ (viewed as a binary relation) if for every $x \in \{0,1\}^n$ there exists $y \in R_x$ such that $\Pr_A[A(x) = y] \geq 2/3$. In the case of $\mathsf{BPP}$-$\mathsf{search}$ and $\mathsf{BPE}$-$\mathsf{search}$, this necessarily means that the solution $y$ produced by $A$ on $x$ is accepted by the verification algorithm $B$ with probability at least $2/3$. We also consider pseudodeterministic algorithms $A$ that only succeed on average with respect to a distribution $\mathcal{D}_n$ supported over $\{0,1\}^n$ and $x \sim \mathcal{D}_n$. In this case, we stress that $A$ is still pseudo-deterministic on \emph{every} input string $x \in \{0,1\}^n$, meaning that it produces a canonical output $z_x$ with probability at least $\geq 2/3$. However, it is not necessarily the case that $z_x \in R_x$ for every input string $x$. These definitions are extended to the infinitely often setting in the natural way. Again, we assume a pseudo-deterministic output for every input string, although the algorithm might not generate a valid solution on some input lengths or on some inputs.\\
	
    We will also rely on the following formalisation of $\mathsf{BPTIME}$-hardness from \citep{DBLP:conf/random/Barak02}.
	
	\begin{definition}[$\mathsf{BPTIME}$-hard problems]\label{d:bptime_hardness} We say that a language $L$ is $\mathsf{BPTIME}$-hard if there is a positive constant $c$ such that, for any time-constructible function $t(n)$ and any language $L' \in \mathsf{BPTIME}[t(n)]$, there is a deterministic $O(t(|x|)^c)$-time computable function $f \colon \{0,1\}^* \to \{0,1\}^*$ such that for every $x$ it holds that $x \in L'$ if and only if $f(x) \in L$. We say that $L$ is $\mathsf{BPP}$-complete if $L$ is $\mathsf{BPTIME}$-hard and $L \in \mathsf{BPP}$.
	\end{definition}
	
	Note that problems in $\mathsf{BPTIME}[t(n)]$ for a large $t(n)$ can produce larger instances of the hard language. Since the reduction in the definition above is deterministic, if $L$ is $\mathsf{BPP}$-complete and $L \in \mathsf{P}$ then $\mathsf{BPP} \subseteq \mathsf{P}$, as one would expect. 
	
	Finally, we introduce notation for the Circuit Acceptance Probability Problem ($\mathsf{CAPP}$). For convenience, we employ a parameter $n$ to index instances. This will be useful when discussing algorithms solving $\mathsf{CAPP}$ on average with respect to an ensemble of distributions.
	
	\begin{definition}[$\mathsf{CAPP}_{n,n^d}$]\label{d:CAPP}
		For a positive integer $d$, we define $\mathsf{CAPP}_{n,n^d}$ to be the search problem where given as input $x = (1^n,C)$, where $|C|=n^d$ and $C$ is interpreted as a Boolean circuit on at most $n^d$ input variables and of size at most $n^d$, we must output a value $\mu \in [0,1]$ such that
		\[
		\left|\Prob_{y\in\bool^{n^{d}}}[C(y)=1]-\mu \right|\leq 1/10.
		\]
	We also define $\mathsf{CAPP}_n \eqdef \mathsf{CAPP}_{n,n \cdot (\log n)^C}$ for a large enough constant $C \geq 1$, which refers to circuits of size $n \cdot (\log n)^C$ and is useful in the context of linear-time probabilistic algorithms.\footnote{By a standard padding argument, the size of the circuits in the definition of $\mathsf{CAPP}$ is not essential, but it is convenient to fix an appropriate size when discussing time bounds and ensembles of input distributions.}
	\end{definition}

    As alluded to above, we consider algorithms solving $\mathsf{CAPP}$ in the worst case and on average, i.e., with respect to an ensemble $\{\mathcal{D}_n\}_{n \geq 1}$ of distributions where each $\mathcal{D}_n$ is supported over $1^n \times \{0,1\}^{n^d}$. When discussing pseudodeterministic algorithms for solving $\mathsf{CAPP}$ on average or in the infinitely often regime, we adopt the same convention as in the case of  $\mathsf{BPP}$-$\mathsf{search}$: the algorithm is assumed to produce with probability at least $2/3$ a canonical value $\mu_x$ on every input string $x = (1^n, C)$, but $\mu_x$ might be incorrect (i.e.~$1/10$-far from $\Pr_y[C(y)=1]$) on some input strings.
	
	Recall that a Turing machine running in time $T$ can be simulated by a Boolean circuit of size $O(T \cdot \log T)$ (see e.g.~\citep{book_complexity}), and that the conversion from machines to circuits can be done efficiently. We will implicitly use this in a few proofs.

	\subsection{Structural properties of languages}

	\begin{definition}[Downward self-reducible language] \label{d:dsr}
		A language $L \subseteq \{0,1\}^*$ is said to be downward self-reducible \emph{(}$\mathsf{dsr}$\emph{)} if there is a polynomial-time oracle algorithm $D$ that for any input $x$, only asks queries of length $< |x|$, and such that $D^L$ decides $L$.
	\end{definition}
	
	\begin{definition}[Paddable language]\label{d:paddable}
		A language $L \subseteq \{0,1\}^*$ is said to be paddable if there is a polynomial-time computable function $f\colon \{0,1\}^* \times \{0,1\}^* \rightarrow \{0,1\}^*$ such that for each $x \in \{0,1\}^*$ and $m > |x|$, $|f(x,1^m)| = m$ and $x \in L$ iff $f(x,1^m) \in L$. 
	\end{definition}
	
	\begin{definition}[Self-correctable language] \label{d:selfcorrectable}
		Let $L \subseteq \{0,1\}^*$ be a language, $C$ be a probabilistic polynomial-time oracle algorithm, and $\varepsilon: \mathbb{N} \rightarrow [0,1]$ be a function. We say that $C$ is an $\varepsilon(n)$ self-corrector for $L$ at input length $n$ if:
		\begin{enumerate}
			
			\item On any input $x$ of length $n$ and for any oracle $O$, $C^O$ only makes queries of length $n$ on input $x$.
			
			\item For all $n \in \mathbb{N}$ and all $O \subseteq \{0,1\}^*$ such that $O(y) = L(y)$ for at least a $1-\varepsilon(n)$ fraction of inputs $y$ of length $n$, $C^{O}(x) = L(x)$ with probability at least $3/4$ \emph{(}over the internal randomness of $C$\emph{)} for each $x$ of length $n$.
			
		\end{enumerate} 
		We say that $L$ is self-correctable if there is a constant $k$ and a probabilistic polynomial-time oracle algorithm $C$ such that $C$ is a $1/n^k$ self-corrector for $L$ at length $n$ for every $n \in \mathbb{N}$.
	\end{definition}
	
	\begin{definition}[Instance-checkable language] \label{d:instancecheckable}
		A language $L$ is said to be same-length instance-checkable if there is a probabilistic polynomial-time oracle machine $I$ with output in $\{0,1,?\}$ such that for any input $x$:
		\begin{enumerate}
			
			\item $I$ only makes oracle queries of length $|x|$.
			
			\item $I^L(x) = L(x)$ with probability $1$.
			
			\item $I^{O}(x) \in \{L(x), ?\}$ with probability at least $2/3$ for any oracle $O$.
			
		\end{enumerate}
		
	\end{definition}
	
	\subsection{Pseudorandomness}

We say that a Boolean function $f \colon \{0,1\}^m \to \{0,1\}$ $\delta$-\emph{distinguishes} distributions $\mathcal{D}_1$ and $\mathcal{D}_2$ supported over $\{0,1\}^m$ if
$$
\Big | \Pr_{y \sim \mathcal{D}_1}[f(y) = 1] - \Pr_{y \sim \mathcal{D}_2}[f(y) = 1] \Big | > \delta.
$$
We will often be interested in the distribution induced by a ``generator'' $G \colon \{0,1\}^{\ell} \to \{0,1\}^m$, by which we mean the distribution $G(\mathcal{U}_\ell)$ supported over $\{0,1\}^m$.

	\begin{theorem}[\cite{DBLP:journals/jcss/ImpagliazzoW01,DBLP:journals/cc/TrevisanV07}]\label{t:io-PRG}
		For every $b \geq 1$, there is a sequence $\{G_\ell\}_{\ell \geq 1}$, where $G_{\ell}\colon \bool^\ell \to \bool^{\ell^{b}}$ is computable in time $2^{O(\ell)}$, such that if there is a polynomial-time samplable distribution $\{\mathcal{D}_{\ell^{b}}\}$ of Boolean circuits and a constant $c$ for which for all sufficiently large $\ell$, with probability at least $\ell^{-b\cdot c}$ over $C \sim \mathcal{D}_{\ell^{b}}$, $C$ $(1/10)$-distinguishes $G_\ell$ from $\mathcal{U}_{\ell^b}$, then  $\mathsf{PSPACE}\subseteq \mathsf{BPP}$.
	\end{theorem}

    We say that a sequence $G = \{G_n\}_{n \geq 1}$ of functions $G_n \colon \{0,1\}^{\ell(n)} \to \{0,1\}^{m(n)}$ is a pseudorandom generator (PRG) against $\mathsf{DTIME}[T]$ with error $\varepsilon(n)$ if for every deterministic algorithm $A$ running in time $T(m)$ on inputs of length $m$, we have for every large enough $n$ that
    $$
    \Big | \Pr_{z \sim \mathcal{U}_{m(n)}}[A(z) = 1] - \Pr_{y \sim G_n(\mathcal{U}_{\ell(n)})}[A(y) = 1]  \Big | \leq \varepsilon(n).
    $$
	We say that $G$ as above is an \emph{infinitely often} pseudorandom generator when for each fixed algorithm $A$ this is only guaranteed to hold for infinitely many values of the parameter $n$. We say that $G$ is computable in \emph{pseudo-deterministic} polynomial time if there is a randomized algorithm $B$ that, when given $1^n$ and $x \in \{0,1\}^{\ell(n)}$, runs in time $\mathsf{poly}(n)$ and outputs $G_n(x)$ with probability at least $2/3$. Finally, we also extend this definition to the case where the randomized algorithm $B$ requires an advice string of length $a(n)$ to compute $G_n$, meaning that there is a function $\alpha \colon \mathbb{N} \to \{0,1\}^*$ with $|\alpha(n)| = a(n)$ such that $B(1^n,x,\alpha(n)) = G_n(x)$ with probability $\geq 2/3$. Note that in this case $B$ does not need to satisfy the promise of bounded acceptance probabilities if it is given an incorrect advice string.

	\subsection{Time-bounded Kolmogorov complexity}\label{s:prelim_kolmogorov}

   We consider natural probabilistic analogues of standard notions from Kolmogorov complexity. We refer the reader to \citep{allender1992applications, DBLP:conf/fsttcs/Allender01, fortnow2004kolmogorov} for more background in time-bounded Kolmogorov complexity and its applications.
	
	We start with the definition of $\rKt$ complexity \citep{DBLP:conf/icalp/Oliveira19}. 	Recall that probabilistic Turing machines have an extra tape with random bits. We will use $\bm{M_{\leq t}}(a)$ to refer to a random variable representing the content of the output tape of $M$ after it computes for $t$ steps over the input string $a$ (or the final content of the output tape if the computation halts before $t$ steps on a given choice of the random string). Fix a universal Turing machine $U$ capable of simulating probabilistic machines (i.e.,~$U$ has its own random tape). We will abuse notation and use $|M|$ to denote the length of the binary encoding of a machine $M$ with respect to $U$.

	\begin{definition}[$\rKt$ complexity of a string] For $\delta \in [0,1]$ and a string $x \in \{0,1\}^*$, we let $$\rKt_\delta(x) = \min_{M,\,a,\,t} \big \{|M| + |a| + \lceil \log t \rceil \,\mid\, \Pr[\bm{M_{\leq t}}(a) = x] \geq \delta \big \},\vspace{-0.2cm}$$
		where the minimisation takes place over the choice of a probabilistic machine $M$, its input string $a$, and the time bound $t$. The randomized time-bounded Kolmogorov complexity of $x$ is set to be $\randkt(x) \eqdef \rKt_{2/3}(x)$.
	\end{definition}
	
	We also introduce a version of (randomised) time-bounded Kolmogorov complexity that fixes a time bound $t$ for the generation of $x$. While a similar definition for deterministic algorithms has been investigated in several works, to our knowledge, its randomised analogue has not been considered before.

\begin{definition}[$\mathsf{rK}^t$ complexity of a string]
For $\delta \in [0,1]$, a string $x \in \{0,1\}^*$, and a time bound $t$, we let $$\mathsf{rK}^t_\delta(x) = \min_{M,\,a} \big \{|M| + |a|\,\mid\, \Pr[\bm{M_{\leq t(|x|)}}(a) = x] \geq \delta \big \},$$
		where the minimisation takes place over the choice of a probabilistic machine $M$ and an input string $a$. The randomized $t$-time-bounded Kolmogorov complexity of $x$ is set to be $\mathsf{rK}^t(x) \eqdef \mathsf{rK}^t_{2/3}(x)$.
\end{definition}

In this work, we will be interested in $\mathsf{rK}^t$ for a fixed polynomial $t(n) = n^b$ (with respect to $n = |x|$), where $b$ might depend on other parameters depending on the context. We might write $\mathsf{rK}^{\mathsf{poly}}$ in informal discussions.
	
We stress that the (deterministic) Kolmogorov complexity measures $\mathsf{Kt}$ and $\mathsf{K}^t$ have been widely investigated in algorithms and complexity, and $\mathsf{rKt}$ and $\mathsf{rK}^t$ are simply natural probabilistic analogues of these measures.

	\section{A polynomial-time computable pseudodeterministic PRG with 1 bit of advice} \label{s:PRG_construction}
	
	This section establishes our main result (Theorem \ref{t:intro_PRG}) and derives new consequences about the time-bounded Kolmogorov complexity of prime numbers and other objects.  
	
	\subsection{The pseudorandom generator}

	\begin{theorem} \label{t:PseudodetPolytimePRG}
		For each $\varepsilon > 0$ and $c,d \geq 1$, there is an infinitely often pseudorandom generator $G = \{G_n\}_{n \geq 1}$ mapping $n^{\varepsilon}$ bits to $n$ bits that is secure against $\DTIME(n^c)$ with error $1/n^d$ and computable in pseudodeterministic polynomial time with $1$ bit of advice. More generally, $G$ is infinitely often secure against any ensemble $\mathfrak{D} = \{\mathcal{D}_n\}_{n \geq 1}$ of distributions $\mathcal{D}_n$ supported over circuits of size $\leq n^c$ and samplable in time $n^c$, in the sense that for infinitely many $n$, with probability at most $1/n^d$ over $C \sim \mathcal{D}_n$ we have that $C$ $1/n^d$-distinguishes $\mathcal{U}_n$ and $G_n(\mathcal{U}_{n^\varepsilon})$. 
	\end{theorem}

	The remainder of this section will be dedicated to a proof of Theorem \ref{t:PseudodetPolytimePRG}. For simplicity, we consider an arbitrary $\varepsilon > 0$ and fix $c =  d = 1$. It is not hard to see that our argument generalises to arbitrary constants $c, d \geq 1$. Moreover, we focus on the case of distinguishers from $\DTIME(n)$. The security of $G$ against samplable circuits follows from a standard adaptation of the proof.
	
	Our construction will use a $\PSPACE$-complete language with certain special properties. This construction is given by \cite{DBLP:conf/focs/Chen19}, building on \cite{DBLP:journals/cc/TrevisanV07}. Chen only claims that the language is self-correctable in a non-uniform sense (as that is all he needs in his proof), but it is clear from his proof of self-correctability that it holds in a uniform sense as well.
	
	\begin{lemma} [\cite{DBLP:journals/cc/TrevisanV07, DBLP:conf/focs/FortnowS04, DBLP:conf/focs/Chen19}] \label{l:splPSPACE}
		There is a $\mathsf{PSPACE}$-complete language $L_{hard}$ that is downward self-reducible, self-correctable, paddable and same-length instance-checkable.
	\end{lemma}
	
	We first show that if $L_{hard}$ can be solved efficiently, we get a much stronger version of Theorem \ref{t:PseudodetPolytimePRG}.
	
	\begin{lemma} \label{l:PSPACEeasy}
		If $L_{hard} \in \BPP$, then there is a \emph{PRG} $\{G_n\}$ with seed length $O(\log(n))$ secure against $\DTIME(n)$, and computable in pseudodeterministic polynomial time.
	\end{lemma}
	
	\begin{proof}
		If $L_{hard} \in \BPP$, then since $L_{hard}$ is $\PSPACE$-complete, we have that $\PSPACE = \BPP$. It follows by a simple padding argument that $\DSPACE(2^{O(n)}) \subseteq \BPE$. By direct diagonalization, there is a language $L'$ in $\DSPACE(2^{O(n)})$ that, for all but finitely many input lengths, does not have circuits of size $2^{0.9 n}$, and by the simulation in the previous sentence, we have that $L' \in \BPE$. Now the desired conclusion follows from Lemma 1 in \cite{DBLP:conf/stoc/OliveiraS17}.
	\end{proof}
	
	Hence we can focus on the case that $L_{hard} \not \in \BPP$. Roughly speaking, we can use the same-length checkability of $L$ to define an {\it optimal} algorithm for $L$, which implies that there is a time bound $T$ for computing $L$ probabilistically that is optimal to within polynomial factors. 
	
	Let $t \colon \mathbb{N} \to \mathbb{N}$ be an arbitrary function. It will be convenient to introduce the following variant of the class $\mathsf{BPTIME}[t]$. We use $\widetilde{\BPTIME}(t(n))$ to denote the set of languages $L$ that admit a probabilistic algorithm $A$ with the following guarantees. For any large enough input length $n$ and for every $x \in \{0,1\}^n$, with probability at least $1 - 1/n$ over its internal randomness $A(x)$ runs for at most $t(n)$ steps and outputs the correct answer $L(x)$. The difference compared with the standard definition $\mathsf{BPTIME}[t]$ is that $A$ might run for more than $t(n)$ steps on some computation paths. Note that when $t$ is \emph{time constructible} the two definitions essentially coincide, since we can always halt the computation of $A$ after $t(n)$ steps. In particular, using constructible upper bounds on running time we have $\widetilde{\mathsf{BPP}} = \mathsf{BPP}$.
	
	For convenience, we say that a function $t \not \in O(\poly(n))$ if for every constant $c \in \mathbb{N}$, there are infinitely many values of $n$ such that $t(n) > c \cdot n^c$.
	
	\begin{lemma} [Adaptation of \cite{DBLP:conf/focs/FortnowS04}] \label{l:optimal}
		Suppose $L_{hard} \not \in \BPP$. There is a non-decreasing function $T \colon \mathbb{N} \rightarrow \mathbb{N}$ and a constant $\delta > 0$ such that for any constant $b>0$, $L_{hard} \in \widetilde{\BPTIME}(T(n)) \setminus \BPTIME(n^b\cdot T(n)^{\delta})/\delta \log T(n)$, and such that $T(n) \not \in O(\poly(n))$.
	\end{lemma}
	\begin{proof}
	The proof described here is similar to the argument in \cite{DBLP:conf/focs/FortnowS04}. The difference is that we use the more convenient $\PSPACE$-complete language stated above, which in fact simplifies the argument. We note that in our presentation we will not \emph{explicitly} state and prove the optimality of the proposed algorithm for $L_{hard}$, as this is not really needed. The lower bound part of the argument relies instead on the definition of the function $T$.
	
	First we describe an algorithm and a corresponding (non-decreasing) function $T \colon \mathbb{N} \to \mathbb{N}$. This function will serve as an upper bound to the running time of the algorithm solving $L_{hard}$ (in the sense of $\widetilde{\mathsf{BPTIME}}[\cdot]$). Let $I$ be an instance checker for $L_{hard}$ with exponentially small error probability $2^{-|x|^c}$, where $c$ is a large enough constant that depends only on $L_{hard}$.
	
			\begin{algorithm}
		\caption{An optimal algorithm for $L_{hard}$}
		\begin{algorithmic}[1]
			
			\Procedure{OPTIMAL}{$x$} 
			
			\For{$m = 1,2,\dots$}	
			\For{each probabilistic program $M$ of description length $\log m$}				
			\State Run $I$ on $x$ with oracle $M^m$ (i.e., $M$ restricted to $m$ steps). 
			\State If a non-``?'' answer $val$ is returned, output $val$.

			\EndFor
			\EndFor
			\EndProcedure
		\end{algorithmic}
	\end{algorithm}
	
	We claim that the above algorithm solves $L_{hard}$ with high probability. Since $L_{hard}\in \mathsf{PSPACE}$, there exists a deterministic exponential-time machine $M_{L_{hard}}$ that decides $L_{hard}$. In a worst-case scenario, such a machine will eventually be tried by the algorithm at some stage $m$ (where $m$ is exponential in $|x|$), and the correct answer will be returned if $M_{L_{hard}}$ is used as an oracle for the instance checker. Also, since the instance checker has exponentially small error probability, the probability that a wrong answer is output before this stage is very small. 
	
	Let $T$ be the non-decreasing function defined as $T(n)=t$, where $t$ is the minimum number such that for every $1\leq i\leq n$, the algorithm OPTIMAL, on inputs of length $i$, outputs the correct answer within $t$ steps, with probability at least $1 - 1/i$. Then we have  $L_{hard} \in \widetilde{\BPTIME}(T(n))$. Note that, since we assume $L_{hard} \not \in \BPP$, it is the case that $T(n) \not \in \poly(n)$. Moreover, the function $T$ is \emph{non-decreasing} by definition. 
	
	Next, we show that $L_{hard} \not\in \BPTIME\left(n^b\cdot T(n)^{\delta}\right)/\delta\log(T(n))$ for any constant $b>0$ and $0<\delta<1/18$. For the sake of contradiction, suppose there are constants $b>0$, $0<\delta<1/18$, and some probabilistic program $M_0$ of size $\delta\log(T(n)) + O(1)$ such that for every input $x$ of length $n$, $M_0$ outputs the correct answer $L_{hard}(x)$ within $n^b\cdot T(n)^{\delta}$ steps, with probability at least $1-1/n$. Then by hardwiring the running time $n^b\cdot T(n)^{\delta}$, we can implement a ``timeout'' mechanism and perform error reduction using standard techniques. Therefore, for every large enough input length $n$, there is a probabilistic program $M$ of size $2\delta\log (T(n)) + b\log(n) + O(1)$ that when restricted to inputs of length $n$ produces the correct answer within $T(n)^{\delta} \cdot n^a$ steps except with exponentially small probability, where $a>b$ is a constant. If the algorithm OPTIMAL reaches stage $m :=T(n)^{2\delta} \cdot n^a$, it will eventually try the program $M$, which has size at most
	\[
		2\delta\log (T(n)) + b\log(n) + O(1) \leq \log m = 2\delta\log(T(n)) + a\log(n).
	\]
	Thus, using $M$ as an oracle for the instance checker, the algorithm outputs the correct answer with high probability. Again, since the instance checker has exponentially small error probability, the probability that a wrong answer is output before this stage is very small. As a result, for every large enough $n$, the algorithm OPTIMAL outputs the correct answer with high probability within $t(n) := m^3\cdot n^c=T(n)^{6\delta}\cdot n^{3a+c}$ steps (to complete stage $m$), for some constant $c>0$. Next, we show the following 
	\begin{claim}\label{c:monotone-via-pad}
		For every large enough $n$,
		\[
		t(n) \geq T(n)^{1/3}/n^d,
		\]
		where $d>0$ is some constant.
	\end{claim}

	\begin{proof}[Proof of \Cref{c:monotone-via-pad}]
		For every $n$, let $s(n)$ be the minimum number of steps such the algorithm OPTIMAL, on inputs of length $n$, outputs the correct answer with probability at least $1 - 1/n$. Note that $t(n) \geq s(n)$ for every $n$. Then to show the claim, it suffices to show that $T(n)\leq s(n)^3\cdot n^d$ for every $n$. Note that by definition, $T(n)=\max_{1\leq i\leq n} s(i)$. Assume without loss of generality that $T(n)=s(\ell)$ for some $\ell\leq n$. To conclude the argument, we argue that $s(\ell)$ is at most $\mathsf{poly}(s(n),n)$ using the \emph{paddability} of $L_{hard}$ and the definition of algorithm OPTIMAL. We give the details below.

		Let $M$ be the following algorithm: on input $x$ of length $\ell$, and an advice encoding the integer $n$, $M$ first computes $x' := p(x,n)\in\bool^n$, where $p$ is the padding function for $L_{hard}$. Then $M$ runs the algorithm OPTIMAL on $x'$.  By the paddability of $L_{hard}$ and the fact that the algorithm OPTIMAL computes $L_{hard}(x')$ within $s(n)$ steps with high probability, we get that $M$ computes $L_{hard}(x)$ within $m_0 := \poly(n) + s(n)$ steps with high probability. Therefore, using error reduction if necessary, we get that for inputs of length $\ell$, there is a program of size $\log(n) + \log (m_0) + O(1)$ that decides $L_{hard}$ within $m := \poly(n)\cdot m_0$ steps with very high probability. Since this program can be used in the algorithm OPTIMAL for inputs of length $\ell$, we conclude that $s(\ell)\leq m^3\cdot n^c$, which implies $T(n) = s(\ell)\leq s(n)^3\cdot n^d$ for some constant $d>0$. This completes the proof of the claim.
	\end{proof}

	\Cref{c:monotone-via-pad} implies that for every large enough $n$,
	\[
		T(n)^{6\delta}\cdot n^{3a+c} \geq  T(n)^{1/3}/n^d,
	\]
	which means
	\[
		T(n) \leq n^{(9a+3c+3d)/(1-18\delta)}.
	\]
	This contradicts that $T(n) \not \in O(\poly(n))$.
	\end{proof}

	Next we argue that a padded version of $L_{hard}$ gives a hierarchy for $\BPP$ with one bit of advice. The argument here is essentially the same as that in Lemmas 14 and 15 in \cite{DBLP:conf/focs/FortnowS04}. The only difference is that we define the padded version slightly differently than in \cite{DBLP:conf/focs/FortnowS04} with a view towards the next part of our proof, but this does not really change the argument.
	
	We define the language $L_{k}$ as follows, where $T$ and $\delta$ are as  in the statement of Lemma \ref{l:optimal}: \\

	\vspace{-0.15cm}
	
	\noindent \textbf{Definition of $L_k$:}\\
$x \in L_{k}$ iff $x = yz$, where $y \in L_{hard}$, $|z| = 2^{\ell}$ for some integer $\ell$, $|z| > |y|$ and $|z| \geq T(i)^{\delta/3k}$ for each non-negative integer $i \leq |y|$. 

	\vspace{0.15cm}
	
	\begin{lemma} [Adaptation of \cite{DBLP:conf/focs/FortnowS04}] \label{l:hierarchy}
		Suppose $L_{hard} \not \in \BPP$. Then $L_k \in \BPP/1 \,\setminus\, \BPTIME(n^k)/1$, for every constant $k>0$.
	\end{lemma}
   	\begin{proof}
		Again, the proof is similar to the argument in \cite{DBLP:conf/focs/FortnowS04}. Firstly we show that $L_k \in \widetilde{\BPP}/1$, by constructing a machine $M$ that takes one bit of advice and with high probability runs in polynomial time and decides $L_k$. This implies that $L_k \in \BPP/1$ by using a constructive upper bound for this regime of time complexity.

		We first specify the sequence of advice bits for $M$.
		We say that input length $m \in \mathbb{N}$ is good for $L_k$ if $m = r + 2^{\ell}$ for non-negative integers $r$ and $\ell$, $m > 2r$ and $2^{\ell} \geq T(i)^{\delta/3k}$ for each $0 \leq i \leq r$.  Note that for $m$ that is good for $L_k$, $r = r(m)$ and $\ell = \ell(m)$ are well-defined, since there is at most one way that any integer $a$ can be written as a sum of non-negative integers $b$ and $c$ such that $c$ is a power of two and $a > 2b$.	For input length $m$, we let the corresponding advice bit $b_m=1$ iff $m$ is good. On input $x$ of length $m$, the machine $M$ rejects $x$ immediately if $b_m=0$. If $b_m=1$, $M$ parses its input as $x=yz$ and accepts if and only if the algorithm for $L_{hard}$ granted by  \Cref{l:optimal} accepts $y$ within $m^{3k/\delta}$ steps. It is clear that $M$ runs in time $\poly(m)$ with high probability. To argue correctness, note that by the definition of $L_k$, if $m$ is not good, then every input of length $m$ is not in $L_k$. Also, if $m$ is good, an input of the form $x=yz$ is in $L_k$ if and only if $y\in L_{hard}$. Then the correctness of $M$ follows from the fact that the algorithm provided by \Cref{l:optimal} takes time $T(|y|)\leq |z|^{3k/\delta}<m^{3k/\delta}$ to output $L_{hard}(y)$ with high probability.
		
		Next, we show that $L_k\not\in \BPTIME(m^k)/1$. For the sake of contradiction, suppose there is a probabilistic machine $M$ that takes one bit of advice and decides $L_{k}$ on inputs of length $m$ in time $m^k$ with high probability. We will construct a probabilistic machine $M'$ that takes $\delta\log(T(n))$ bits of advice and decides $L_{hard}$ on inputs of length $n$ in time $\poly(n)\cdot T(n)^{\delta}$ with high probability, which contradicts \Cref{l:optimal}.
		Given an input $y$ of length $n$ for $L_{hard}$, $M'$ interprets the first part of its advice as an encoding of the smallest integer $\ell$ such that $2^{\ell}>n$ and $2^{\ell}\geq T(n)^{\delta/3k}$, and obtains a padded input $x=y1^{2^{\ell}}$. Since $T$ is non-decreasing, we also get that $2^{\ell}\geq T(i)^{\delta/3k}$ for each $i \leq n$, which means that the input length of $x$ is good. Then $M'$ interprets the second part of its advice as the correct advice bit for $M$ and it accepts if and only if $M$ accepts the padded input $x$ with this advice bit. Note that the number of advice bits for $M'$ is at most $\delta\log(T(n))/(3k)+O(1)\leq \delta \log(T(n))$. Also since $|x|\leq n +  2\cdot T(n)^{\delta/3k}$, $|x|^k\leq (2\cdot n)^k \cdot  T(n)^{\delta/3}$. Therefore, $M$ decide $L_{hard}$ on inputs of length $n$ within  $O(2\cdot n)^k \cdot  T(n)^{\delta/3}$ steps with high probability, which contradicts \Cref{l:optimal}.
	\end{proof}

	Now we get to the core of our proof: plugging in the language $L_k$ for appropriately chosen $k$ into a version of the Nisan-Wigderson generator, and arguing that the resulting PRG is secure against uniform adversaries. This involves using a learning procedure that is specifically tailored to the structure of the language $L_k$.
	
	The following lemma is stated slightly differently than {\cite[Lemma 3.5]{DBLP:journals/cc/TrevisanV07}}, but the proof is exactly the same.
	
	\begin{lemma} [\cite{DBLP:journals/jcss/ImpagliazzoW01,DBLP:journals/cc/TrevisanV07}] \label{l:learnablePRG}
		Let $L$ be a language, $C$ be a probabilistic polynomial-time oracle algorithm, $D$ be a polynomial-time algorithm, and let $\varepsilon > 0$ be any constant. There is a generator $G = \{G_n\}_{n \geq 1}$ with seed length $n^{\varepsilon}$ and producing $n$ output bits such that:
		
		\begin{enumerate}
			
			\item[\emph{(\emph{i})}] \emph{Complexity:} $G_n$ can be computed in polynomial time given oracle access to $L$ on inputs of length  $m(n) = n^{\gamma}$, for some $\gamma < \varepsilon$.
			
			\item[\emph{(\emph{ii})}] \emph{``Exact Learnability'':} For every constant $a > 0$, there is a probabilistic polynomial-time oracle algorithm $B$ with unary input such that for each $n$ for which $D$ $1/n$-distinguishes the output of $G_n$ from random, and for which $C$ is a $1/m^a$ self-corrector for $L$ at length $m = n^{\gamma}$, $B(1^{m})$ makes oracle queries to $L$ of length exactly $m$, and with probability at least $1-1/n^2$ outputs a circuit $Ckt$ that correctly computes $L$ at length $m$.
			
		\end{enumerate}
		
	\end{lemma}

	We apply Lemma \ref{l:learnablePRG} to the language $L_k$ (for $k$ to be determined later) to obtain the generator $\{G_n\}$ in Theorem \ref{t:PseudodetPolytimePRG}. Next, we show that $\{G_n\}$ is computable in pseudodeterministic polynomial time with 1 bit of advice, and that it is secure infinitely often against $\DTIME(n)$ adversaries.\\

\noindent \textbf{Complexity of computing $G_n$.} The computability condition is much easier to establish. Let $M$ be an advice-taking probabilistic machine deciding $L_k$ in polynomial time with one bit of advice and with error $1/n^{\omega(1)}$. We define an advice-taking probabilistic polynomial-time machine $M'$, which given an input $x$ of length $n^\varepsilon$ and one bit of advice, computes $G_n(x)$ pseudodeterministically. $M'$ simulates the polynomial time oracle procedure given by the first item of Lemma \ref{l:learnablePRG}, and each time the oracle procedure makes a query of length $n^{\gamma}$, $M'$ runs $M$ with the correct advice bit for length $n^{\gamma}$ to answer the query. Since $M$ runs in polynomial time, $M'$ runs in polynomial time. To see that $M'$ is pseudodeterministic, note that the oracle procedure makes at most $\poly(n)$ queries, since it runs in polynomial time, and by a union bound over the random choices of $M$, all of these queries are answered correctly with probability $1-1/n^{\omega(1)}$. Hence with probability $1-1/n^{\omega(1)}$, $M'$ outputs $G_n(x)$ correctly.\\ 
	
	\noindent \textbf{Security of $G_n$.}~In order to argue that $\{G_n\}$ is secure against $\DTIME(n)$ adversaries for infinitely many $n$, we use the learning procedure in the second item of Lemma \ref{l:learnablePRG} in conjunction with structural properties of the language $L_k$ (which is defined using the special language $L_{hard}$). This argument is somewhat technical, and we establish some new terminology first. The definition given below appears in the proof of Lemma \ref{l:hierarchy}, but we present it again in case the reader skipped that argument.\\

	\noindent \emph{Good input length.} A key notion is that of a {\it good} input length $m$ for $L_k$. We say that input length $m \in \mathbb{N}$ is good for $L_k$ if $m = r + 2^{\ell}$ for non-negative integers $r$ and $\ell$, $m > 2r$ and $2^{\ell} \geq T(i)^{\delta/3k}$ for each $0 \leq i \leq r$.  Note that for $m$ that is good for $L_k$, $r = r(m)$ and $\ell = \ell(m)$ are well-defined, since there is at most one way that any integer $a$ can be written as a sum of non-negative integers $b$ and $c$ such that $c$ is a power of two and $a > 2b$. By the definition of $L_k$, if $m$ is not good for $L_k$, then every input of length $m$ is not in $L_k$. (While we won't explicitly rely on this, as a sanity check note that for each $r$ there are large enough integers $\ell$ and $m$ such that $r = r(m)$, $\ell = \ell(m)$, and $m$ is good for $L_k$.)
	
	For each good input length $m$, we define an increasing sequence $I_m = m_0, \ldots , m_{r(m)}$ as follows: $m_i = i + 2^{\ell(m)}$. We argue that for each $0 \leq i \leq r(m)$, $m_i$ is a good input length for $L_k$. The first condition for goodness is clearly satisfied: each $m_i$ can be decomposed as $i$ plus a power of two; moreover, $r(m_i) = i$ and $\ell(m_i) = \ell(m)$. Also, since $m > 2 r(m)$, it follows that $m > 2i$ for each $i \leq r(m)$. Finally, since $2^{\ell(m)} \geq T(i)^{\delta/3k}$ for each $i \leq r(m)$, we have that for each $m_i$, $2^{\ell(m_i)} = 2^{\ell(m)} \geq T(i)^{\delta/3k}$ for each $j \leq r(m_i) = i \leq r(m)$. Intuitively, each $m_i$ in the sequence $I_m$ inherits its goodness from $m$.\\ 
	
	We will use good input lengths and their corresponding sequences in 2 ways: first, we use the self-correctability of $L_{hard}$ to give a probabilistic polynomial-time oracle procedure that is a self-corrector for $L_k$ on each good input length, and second, we use the downward self-reducibility of $L_{hard}$ to argue that if $L_k$ is learnable on good input lengths, then there is a probabilistic polynomial-time machine $N$ with one bit of advice deciding $L_k$ everywhere. The one bit of advice for $N$ will be used to tell if an input length is good for $L_k$.

	\begin{lemma} \label{l:hierselfcorr}
		There is a probabilistic polynomial-time oracle procedure $C$  and a constant $a > 0$ such that for each good input length $m$ for $L_k$, $C$ is a $1/m^a$ self-corrector for $L_k$ at length $m$.
	\end{lemma}      
	
	\begin{proof}
		By assumption, $L_{hard}$ is self-correctable, and therefore there is a constant $b > 0$ and a probabilistic polynomial-time oracle algorithm $C'$ such that $C'$ is a $1/n^b$ self-corrector for $L_{hard}$ with success probability $3/4$. We define $C$ as follows. Given input $x$ of length $m$ and access to an oracle, it checks if $m = r + 2^{\ell}$ for non-negative integers $r$ and $\ell$ with $m > 2r$. This check can easily be implemented in polynomial time. If the check fails, $C$ rejects. If the check succeeds, let $x = yz$, where $|y| = r$ and $|z|$ is a power of two. $C$ simulates $C'$ in the following way. It runs $C'$ on $y$. Whenever $C'$ makes an oracle query $y'$ of the same length as $y$, $C$ makes oracle queries to $y' z_1, \ldots, y' z_{100 m}$ where each $z_i$ is chosen uniformly at random from strings of length $|z|$, and uses the majority answer of these queries as the simulated answer to $y'$. (If $C'$ queries the same input twice, $C$ provides a consistent answer.) $C$ accepts its input string $x$ iff the above simulation involving $C'$ accepts.
		
		We argue that $C$ is a $1/m^a$ self-corrector for $L_k$ at any good length $m$, where $a = 2b$. By the definition of $L_k$, if $m$ is a good length, then $x$ of length $m$ belongs to $L_k$ iff the $r(m)$ length prefix $y$ of $x$ belongs to $L_{hard}$. Suppose that $O$ is an oracle that agrees with $L_k$ on at least a $1-1/m^a$ fraction of inputs of length $m$. We show that $C^O$ decides $L_k$ correctly on $x$ for each $x$ of length $m$. Call a string $y'$ \emph{nice} if for at least a $2/3$ fraction of strings $z'$ of length $|z|$, $y'z' \in O$ iff $y' \in L_{hard}$. By a straightforward application of the Markov bound, at least $1-3/m^a$ fraction of strings of length $r$ are nice. Define the partial oracle $O'$ at length $r$ by setting $O'(y') = L_{hard}(y')$ if $y'$ is nice. $O'(y')$ is left undefined for strings $y'$ that are not nice. By the lower bound on fraction of nice strings of length $r$, $O'$ is defined for at least $1-3/m^a \geq 1-1/r^b$ fraction of strings of length $r$, since $m > 2r$ and $a = 2b$.
		
		By a simple Chernoff bound and a union bound, for every string $y'$ on which $O'$ is defined, with all but exponentially small probability, the simulation by $C$ of an oracle query $y'$ of $C'$ returns $L_{hard}(y')$. Since $C'$ is a $1/n^b$ self-corrector for $L_{hard}$ and the partial oracle $O'$ is defined and agrees with $L_{hard}$ for at least a $1-1/r^b$ fraction of $r$-bit strings, it follows by convexity that the simulation of $C'(y)$ returns $L_{hard}(y)$ for each $y$ with success probability $3/4 - 2^{-\Omega(m)} \geq 2/3$. Since $L_k(x) = L_{hard}(y)$, this implies that on oracle $O$ the oracle algorithm $C$ outputs $L_k(x)$ with probability at least $2/3$ for each $x$ of length $m$.
	\end{proof}

	We apply Lemma \ref{l:learnablePRG} together with Lemma \ref{l:hierselfcorr} and the downward self-reducibility of $L_{hard}$ to establish that the PRG $\{G_n\}$ is secure against $\DTIME(n)$ infinitely often. Contrapositively, let $D$ be a deterministic linear-time algorithm that $1/n$-distinguishes the output of $G_n$ from random on almost all lengths $n$. We show, for any sufficiently large $k$, that this implies that $L_k$ in $\BPTIME(n^k)/1$, in contradiction to the lower bound in Lemma \ref{l:hierarchy}.
	
	We define an advice-taking probabilistic poly-time machine $N$ with one bit of advice as follows. Given an input $x$ of length $m$, $N$ uses its advice bit to tell if the input length $m$ is good for $L_k$. If the length $m$ is not good, $N$ rejects. If $m$ is good, $N$ inductively builds circuits $\mathsf{Ckt}_i, i = 0, \ldots, r(m)$, where $\mathsf{Ckt}_i$ decides $L_k$ at length $m_i \in I_m$. $\mathsf{Ckt}_0$ is a trivial circuit that is the constant 1 iff the empty string is in $L_{hard}$ and the constant $0$ otherwise. For $i > 0$, $N$ inductively builds $\mathsf{Ckt}_i$ from circuit $\mathsf{Ckt}_{i-1}$ by using the learnability of the generator $\{G_n\}$ and the downward self-reducibility of $L_{hard}$. 
	
	Let $n_i = m_i^{1/\gamma}$. We apply Lemma \ref{l:learnablePRG} to $L_k$, the oracle algorithm $C$ from Lemma \ref{l:hierselfcorr}, the deterministic linear-time algorithm $D$ that $1/n_i$-distinguishes the output of $G_{n_i}$ from random, and the constant $\varepsilon$ in the statement of Theorem \ref{t:PseudodetPolytimePRG}. Using the fact that $m_i$ is good for $L_k$, it follows from Lemma \ref{l:hierselfcorr} that the oracle procedure $C$ is a self-corrector for $L_k$ at length $m_i$. Since the conditions of the second item of Lemma \ref{l:learnablePRG} are satisfied, the probabilistic poly-time oracle procedure $B(1^{m_i})$ on oracle $L_k$ only asks queries of length exactly $m_i$ and outputs a correct circuit $\mathsf{Ckt}_i$ for $L_k$ at length $m_i$. We need to simulate the oracle procedure by a procedure that does not use an oracle, and we do so by taking advantage of the downward self-reducibility of $L_{hard}$.
	
	By Lemma \ref{l:splPSPACE}, the language $L_{hard}$ is downward self-reducible. This means there is a polynomial-time oracle algorithm $A$ that solves $L_{hard}$ on input $x$ while only making queries to $L_{hard}$ on inputs of length less than $|x|$. By induction, we have that the advice-taking probabilistic poly-time machine $N$ has already computed correct circuits $\mathsf{Ckt}_0, \ldots, \mathsf{Ckt}_{i-1}$, where $\mathsf{Ckt}_j$ is a circuit of size $\poly(m_j)$ correctly solving $L_k$ on inputs of length $m_j$. In order to compute a correct circuit $\mathsf{Ckt}_i$ at length $m_i$, $N$ runs $B(1^{m_i})$, answering any oracle query $q$ of $B$ as follows. By definition of $m_i$, $q = q_1 q_2$, where $|q_1| = r(m_i) = i$ and $|q_2| = 2^{\ell(m)}$. Moreover, since $m_i$ is good, $q \in L_k$ iff $q_1 \in L_{hard}$. $N$ runs the downward self-reduction $A$ on $q_1$, generating new queries all of length less than $i$. Let $q'$ be such a query to $L_{hard}$ of length $j$. $N$ constructs circuits for $L_k$ rather than $L_{hard}$, so it simulates the query $q'$ by running the circuit $\mathsf{Ckt}_j$ on $q'w_m$, where $w_m$ is a string of 0s of length $2^{\ell(m)}$. Note that $q' \in L_{hard}$ iff $q' w_m \in L_k$ -- this is because $q' w_m$ is of length $m_j$, which is a good input length. Hence each query of the downward self-reduction $A$ is answered correctly, and moreover so is each query of the learning algorithm $B$. Therefore $N$ correctly produces a circuit $\mathsf{Ckt}_{i}$ for length $m_i$ with high probability at the end of its simulation of $B$. Clearly, the simulation of $B$ runs in polynomial time, and moreover the size of the circuit output by $B$ is a fixed polynomial independent of the complexity of the simulation of the oracle. $N$ returns $\mathsf{Ckt}_{r(m)}(x)$. By a union bound over the $r$ iterative phases of $N$, $B$ outputs a correct circuit with high probability on all phases, and therefore $N$ returns the correct answer for $L_k(x)$.
	
	We need to fix $k$ so as to derive a contradiction. The advice-taking probabilistic algorithm $N$ runs in time $m^c$ for some fixed $c$ that depends only on $L_{hard}$ and the ``learning'' algorithm $B$ (which depends on $D$), and not on $k$. Hence we can simply set $k$ large enough to derive a contradiction to Lemma \ref{l:hierarchy}. \qed

	\subsection{Improved bounds for primes and further applications}\label{s:applications}

	In this section, we show (unconditionally) that dense languages in $\mathsf{P}$ must contain strings of $\mathsf{rK}^{\mathsf{poly}}$ complexity bounded by $n^{\varepsilon}$. We refer the reader to Section \ref{s:prelim_kolmogorov} for definitions related to time-bounded Kolmogorov complexity.
	
	Recall that, for a function $\mu \colon \mathbb{N} \to [0,1]$, we say that a language $L \subseteq \{0,1\}^*$ is $\mu$-\emph{dense} if for every large enough $n$, we have $\Pr_{y \sim \{0,1\}^n}[y \in L] \geq \mu(n)$.
	
	\begin{theorem}\label{t:dense_set_P_bound}
		Let $L \in \mathsf{P}$ be a language of density $\mu(n) \geq 1/n^c$, for some positive constant $c$. Then, for every $\varepsilon > 0$ there is a constant $k\geq 1$ for which the following holds. For infinitely many input lengths $n$, there is a string $x \in \{0,1\}^n$ such that $x \in L$ and $\mathsf{rK}^t(x) \leq n^{\varepsilon}$, where $t = n^k$. 
	\end{theorem}
	
	\begin{proof}
		Let $L \in \mathsf{P}$, i.e., suppose that $L \in \mathsf{DTIME}[n^d]$ for some constant $d$. Take a fixed $\varepsilon > 0$, and consider an infinitely often pseudodeterministic polynomial-time computable PRG $\{G_n\}_{n}$ with $1$ bit of advice given by Theorem \ref{t:io-PRG} with $G_n \colon \{0,1\}^{n^{\varepsilon/2}} \to \{0,1\}^{n}$ that is secure against $\mathsf{DTIME}[n^d]$ and has associated error parameter $\gamma = 1/2n^c$. Since each output of $G_n$ can be computed in polynomial time with high probability assuming  the correct advice bit is given, it is easy to see that for $w \in \{0,1\}^{n^{\varepsilon/2}}$ and $y = G_n(w)$, we have $\mathsf{rK}^t(y) \leq O_G(1) + O(\log n) + 1 + n^{\varepsilon/2} \leq n^{\varepsilon}$, provided that $t = n^k$ for a large enough constant $k$ that is independent of $n$. Moreover, using the density of $L$ and the error parameter of $G$, it follows that for infinitely many choices of the parameter $n$ we have $G_n(\{0,1\}^{n^{\varepsilon/2}}) \cap L \neq \emptyset$. As a consequence, for infinitely many input lengths $n$, there is a string $x \in \{0,1\}^n$ such that $x \in L$ and $\mathsf{rK}^t(x) \leq n^{\varepsilon}$.
	\end{proof}
	
	As an immediate consequence of this theorem, the density of primes, and $\mathsf{Primes} \in \mathsf{P}$ \citep{Agrawal02primesis}, we get that infinitely many prime numbers have bounded $\mathsf{rK}^{\mathsf{poly}}$ complexity. 
	
	\begin{corollary}\label{c:primes}
		For every $\varepsilon > 0$, there is an infinite sequence $\{p_m\}_{m \geq 1}$ of increasing primes $p_m$ such that $\mathsf{rK}^t(p_m) \leq |p_m|^\varepsilon$, where $t(n) = n^k$ for some constant $k = k(\varepsilon) \geq 1$, and $|p_m|$ denotes the bit-length of $p_m$.
	\end{corollary}

	If we interpret the bound $\mathsf{rK}^t(p_m) \leq |p_m|^\varepsilon$ from a data compression perspective, Corollary \ref{c:primes} shows that for infinitely many values of $n$ there are $n$-bit primes that can be decompressed from a representation of length $n^{\varepsilon}$ with high probability and in \emph{polynomial time}. This running time offers an \emph{exponential} improvement compared to the $\mathsf{rKt}$ upper bounds for prime numbers established by \citep{DBLP:conf/stoc/OliveiraS17, DBLP:conf/icalp/Oliveira19}, which provide representation length $n^{\varepsilon}$ but only guarantee decompression (with high probability) in time $2^{n^{\varepsilon}}$.
	
	We can use a similar approach to obtain the following consequence for the problem of generating primes.
	
	\begin{corollary}
	    For every constant $\varepsilon > 0$, there is a probabilistic polynomial time algorithm $A$ with the following property. For infinitely many values of $n$, there exists an $n$-bit prime $p_n$ such that $\Pr_A[A(1^n) = p_n] \geq 2^{-n^\varepsilon}$.
	\end{corollary}
	
	\begin{proof}[Proof Sketch.]
	    The argument is not very different from the proofs of Theorem \ref{t:dense_set_P_bound} and Corollary \ref{c:primes}. For a given $\varepsilon > 0$, we instantiate a pseudo-deterministic PRG $G$ with appropriate parameters in order to fool a deterministic polynomial time algorithm for checking if a given integer is prime. The algorithm $A$ from the statement of the result randomly guesses the advice bit and a seed $w$ of length $n^{\varepsilon/2}$ for $G$, then outputs the string $G(w) \in \{0,1\}^n$ using the pseudo-deterministic algorithm for computing $G$. On infinitely many input lengths where the generator succeeds, with probability at least $(1/2) \cdot 2^{-n^{\varepsilon/2}} \cdot (2/3) \geq 2^{-n^\varepsilon}$ the correct advice bit is generated, the canonical string produced by $G$ on the given seed $w$ is a prime number (since at least one output string of $G$ must represent a prime number), and the pseudo-deterministic algorithm for $G$ produces the canonical output.
	\end{proof}
	
	We prove the following unconditional complexity lower bound, which shows that estimating $\mathsf{rK}^{t}$ up to a polynomial is hard, in the regime where $t$ is larger than the running time of the algorithm trying to estimate $\mathsf{rK}^{t}(x)$ on an input string $x$.

	\begin{theorem}[An unconditional complexity lower bound for estimating {$\mathsf{rK}^{\mathsf{poly}}$}]
	For any $\varepsilon > 0$ and $d \geq 1$ there exists a constant $k \geq 1$ for which the following holds. Consider the following promise problem $\Pi^{\varepsilon} = (\mathcal{YES}_n, \mathcal{NO}_n)_{n \geq 1}$, where
	\begin{eqnarray}
	\mathcal{YES}_n & = & \{x \in \{0,1\}^n \mid \mathsf{rK}^t(x) \leq n^{\varepsilon} \}, \nonumber \\
	\mathcal{NO}_n & = & \{x \in \{0,1\}^n \mid \mathsf{rK}^t(x) \geq n - 1 \}, \nonumber
	\end{eqnarray}
		and $t(n) = n^k$. Then $\Pi^\varepsilon \notin \mathsf{promise}$-$\mathsf{BPTIME}[n^d]$.
	\end{theorem}

		\begin{proof}
		Suppose there is an algorithm $A$ running in probabilistic time $O(n^d)$ that accepts strings in $\mathcal{YES}_n$ and rejects strings in $\mathcal{NO}_n$, where we set $t(n) = n^k$ for a large enough constant $k = k(d,\varepsilon)$. We assume without loss of generality, using amplification if necessary, that the error probability of $A$ on any string from $\mathcal{YES}_n \cup \mathcal{NO}_n$ is at most $2^{-2n}$. Let $C_n^w(x)$ be a Boolean circuit that computes as $A(x)$ on a given $x \in \{0,1\}^n$ when the random input string of $A$ is set to $w \in \{0,1\}^{O(n^d)}$. Note that the collection $\{C_n^w\}$ (for a uniformly random string $w$) can be sampled in time at most $n^{c}$ for some constant $c = c(d)$, and each circuit $C_n^w(x)$ is also of size at most $n^c$. Moreover, by a union bound, with probability at least $1/2$ over the choice of $w$, the (deterministic) circuit $C_n^w$ is correct on every string in $\mathcal{YES}_n \cup \mathcal{NO}_n$. If this is the case, we say that $C_n^w$ is good. 
		
		Now consider the PRG $G = \{G_n\}_{n \geq 1}$ obtained from Theorem \ref{t:PseudodetPolytimePRG} for seed length $n^{\varepsilon/2}$, our parameter $c$, and $d = 1$. Since every output string  $x = G_n(z)$ has $\mathsf{rK}^t$ complexity at most $n^{\varepsilon}$ (if $k$ is large enough), any good circuit $C_n^w$ accepts $x$. On the other hand, since at least half of the $n$-bit strings are in $\mathcal{NO}_n$, we have that $\Pr[C_n^w(\mathcal{U}_n) = 1] \leq 1/2$ for any good circuit $C_n^w$. In other words, a good circuit $1/n$-distinguishes $G_n(\mathcal{U}_{n^{\varepsilon/2}})$ and $\mathcal{U}_n$.
		
		As a consequence of the discussion above, if $k = k(d,\varepsilon)$ is large enough, $\mathcal{D}_n = \{C_n^w\}$ gives rise to a samplable distribution of circuits that break the pseudorandomness of the generator $G_n$, in contradiction to Theorem \ref{t:PseudodetPolytimePRG}. We conclude from this that $\Pi^\varepsilon \notin \mathsf{promise}$-$\mathsf{BPTIME}[n^d]$, which completes the proof.
	\end{proof}
	
	We note that a complexity lower bound for computing $\mathsf{K}^t$ (against deterministic algorithms and for large enough $t$) was recently established by Hirahara \citep{DBLP:conf/stoc/Hirahara20} using different techniques.

	\section{Better pseudo-derandomisations yield new structural results}

	It is well known and easy to show that if we have a polynomial-time \emph{almost-everywhere}  \emph{deterministic} algorithm for $\mathsf{CAPP}$, then $\mathsf{BPTIME}$ admits complete problems. Our main results in this section show that much weaker pseudo-derandomisations of $\mathsf{CAPP}$  would also  have interesting consequences for the structure of probabilistic time. These results formalise the implications informally stated in Theorem \ref{t:pseudo-to-structural}.

	\subsection{Hierarchies from weak pseudo-derandomisations of \texorpdfstring{$\mathsf{CAPP}$}{CAPP}}
	
	In this section, we show that weak pseudo-derandomisations of $\mathsf{CAPP}$ imply hierarchy theorems for probabilistic time.
	
	\begin{theorem}[Pseudo-derandomisation of $\mathsf{CAPP}$ yields probabilistic time hierarchies]\label{t:pd-CAPP-to-hierarchy}
		Let $T$ be a constructive time bound, and let $d\geq 1$ be a constant. If for every polynomial-time samplable ensemble of distributions $\mathcal{D}_{n, n^{d+1}}$ supported over circuits whose description is of length $n^{d+1}$ there is a pseudodeterministic algorithm for $\mathsf{CAPP}_{n,n^{d+1}}$ that runs in time $T(n)$ and succeeds with probability at least $1-1/(3n)$ over $\mathcal{D}_{n, n^{d+1}}$ for infinitely many values of $n$, then there is a language $L \in \mathsf{BPTIME}[T(n)]$ such that $L \notin \mathsf{BPTIME}\!\left[n^d\right]$. Moreover, if the pseudodeterministic algorithm for $\mathsf{CAPP}_{n,n^{d+1}}$ succeeds \emph{(}on average\emph{)} on all sufficiently large $n$, then there is a language $L \in \mathsf{BPTIME}[T(n)]$ such that $L \notin \mathsf{i.o.BPTIME}\!\left[n^d\right]$.
	\end{theorem}
	\begin{proof}
		Let $B_1, B_2, \ldots$ be an enumeration of all (clocked) probabilistic machines running in time $n^d$. Let $A$ be a (i.o.-)pseudodeterministic search algorithm for $\mathsf{CAPP}_{n,n^{d+1}}$ that succeeds  with high probability over a particular polynomial-time samplable distribution over circuits defined below.
		
		We first define the language $L$. 
		Given $x \in \{0,1\}^{n}$, if $x$ is not of the form $1^{n - \lceil\log n \rceil} i$ for some $i \in \{0,1\}^{\lceil\log n \rceil}$, then reject. Otherwise, let $C_{i}(y)$ be the Boolean circuit of size at most $n^{d + 1}$ that computes according to $B_i\!\left(1^{n - \lceil\log n \rceil}i, y\right)$, where $y$ is the internal randomness used by $B_i$. Then we accept $x$ if and only if $A(1^n,C_i) \leq 1/2$.
		
		Since $A$ is a pseudodeterministic algorithm that runs in time $T(n)$, and given $i$ we can easily compute its input circuit $C_i$, we get that $L \in \mathsf{BPTIME}[T(n)]$. 
		
		Next, we show that $L\not\in\mathsf{BPTIME}\!\left[n^d\right]$. Let $L'$ be an arbitrary language in $\mathsf{BPTIME}[n^d]$. Then there is an $i$ such that the machine $B_i$ computes $L'$. Let $n\geq i$ be such that our pseudodeterministic algorithm $A$ succeeds on $n$ when the input circuits coming from the distribution $\mathcal{D}_{n,n^{d+1}}$ defined by sampling a random string $i$ of length $\lceil\log n\rceil$ and computing the circuit $C_i(\cdot)=B_i\!\left(1^{n - \lceil\log n \rceil}i, \cdot\right)$. Note that $\mathcal{D}_{n,n^{d+1}}$ is samplable in polynomial time. Assume without loss of generality that $1^{n - \lceil \log n \rceil}i\in L'$. Then we have 
		\[
		\Prob_y\left[B_i\!\left(1^{n - \lceil\log n \rceil}i, y\right)=1\right]=\Prob_y[C_i(y)=1]\geq 2/3.
		\]
		Note that, for the distribution $\mathcal{D}_{n,n^{d+1}}$ defined above, each element in its support has probability weight at least $1/2^{\lceil\log n\rceil}\geq 1/(2n)$. Since the pseudodeterministic algorithm $A$ succeeds with probability at least $1-1/(3n)$ over this input distribution, we have that $A$ succeeds on every input in its support, including $C_i$. In other words, the canonical output of $A$ on $(1^n,C_i)$ is at least $2/3-1/10>1/2$, which means that $1^{n - \lceil\log n \rceil}i \not\in L$.
		
		It is easy to check that the ``moreover'' part follows from a similar argument.
	\end{proof}
	
	Recall that \cite{DBLP:conf/stoc/OliveiraS17} established the following unconditional (average case, infinitely often, sub-exponential time) pseudo-derandomisation of $\mathsf{CAPP}$ (see \Cref{s:appendix_OS17} for a sketch of the proof).
	
	\begin{theorem}[$2^{n^{\varepsilon}}$-time infinitely often average-case pseudo-derandomisation of $\mathsf{CAPP}_{n,n^d}$ ]\label{t:uncondition-pd-CAPP}
		For any constants $\varepsilon>0$ and $c, d\geq 1$, there is a pseudodeterministic algorithm for $\mathsf{CAPP}_{n,n^{d}}$ that runs in time $2^{O(n^{\varepsilon})}$, and for any polynomial-time samplable ensemble of distributions $\mathcal{D}_{n, n^d}$ supported over circuits of size $\leq n^d$, succeeds with probability $1-1/n^c$ over $\mathcal{D}_{n, n^d}$ for infinitely many values of $n$. 
	\end{theorem}
	
	As a consequence of Theorems  \ref{t:pd-CAPP-to-hierarchy} and \ref{t:uncondition-pd-CAPP}, we get the following corollary, which provides an alternate proof of an existing hierarchy theorem.
	
	\begin{corollary}
		For every constant $k\geq 1$ and each $\varepsilon > 0$, there is a language $L \in \mathsf{BPTIME}[2^{n^{\varepsilon}}] \setminus \mathsf{BPTIME}[n^k]$.
	\end{corollary}

	\subsection{\texorpdfstring{$\mathsf{BPTIME}$}{BPTIME}-hardness from pseudo-derandomisations of \texorpdfstring{$\mathsf{CAPP}$}{CAPP}}
	
	In this section, we show that weak pseudo-derandomisations of $\mathsf{CAPP}$ imply different forms of $\mathsf{BPTIME}$-hardness. 
	
	\begin{theorem}\label{thm:HeurBPP_hard}
	    Let $c, d \geq 1$, and let $T$ be a monotone constructive time bound.
		Suppose that there is a pseudodeterministic algorithm for $\mathsf{CAPP}_{n,n^{d+1}}$ that runs in time $T(n)$, and for every polynomial-time samplable  distribution $\mathcal{D}_{n, n^{d+1}}$ over circuits whose description is of length $\leq n^{d+ 1}$, succeeds with probability $1-1/n^c$ over $\mathcal{D}_{n, n^{d+1}}$ for infinitely many values of $n$. Then there is a language $L \in \mathsf{BPTIME}[T(n)]$ such that, for every language $L_0\in \mathsf{BPTIME}[n^d]$, there is a deterministic polynomial-time reduction $R$ such that, for every polynomial-time samplable distribution $I_n$ supported over $\{0,1\}^n$ and for infinitely many values of $n$, we have
		\[
		\Prob_{x\sim I_n}[L_0(x) = L(R(x))]\geq 1-1/n^c.
		\]
	\end{theorem}
	
	Combining \Cref{thm:HeurBPP_hard} with the unconditional pseudo-derandomisations for $\mathsf{CAPP}$ in \Cref{t:uncondition-pd-CAPP}, we get the following unconditional result.
	
	\begin{corollary}\label{c:new_completeness}
		For every $\varepsilon>0$ and $c, d \geq 1$, there is a language $L\in \mathsf{BPTIME}\!\left[2^{n^{\varepsilon}}\right]$ such that, for each language $L_0\in \mathsf{BPTIME}[n^{d}]$, there is a deterministic polynomial-time reduction $R$ such that, for every polynomial-time samplable distribution $I_n$ supported over $\{0,1\}^n$ and for infinitely many values of $n$, we have
		\[
		\Prob_{x\sim I_n}[L_0(x) = L(R(x))]\geq 1-1/n^c.
		\]
	\end{corollary}
	
	We now prove \Cref{thm:HeurBPP_hard}.
	\begin{proof}[Proof of \Cref{thm:HeurBPP_hard}]
		Given a probabilistic machine $M$ that runs in at most $t = |x|^d$ steps and an input $x$ for $M$, we let $C_{(M,x)}(y)$ be the circuit that computes according to $M(x, y)$, where $y$ is the internal randomness used by $M$. Recall that given $M$ and $x$, $C_{(M,x)}$ is easily computed and has size at most $O(t \cdot \log t) \leq |x|^{d + 1}$.
		
		We now define the language $L$. Let $A$ be the pseudodeterministic search algorithm for $\mathsf{CAPP}_{n,n^{d + 1}}$ granted by the statement of the theorem. For an input $w=\left(\langle M\rangle,x,1^t\right)$ of length $n$ with $t \leq |x|^d$,  
		\[
		w\in L \iff \text{The canonical output $\mu$ of $A\!\left(1^{|x|}, C_{(M,x)}\right)$ is at least 1/2}.
		\]
		Since $A$ is a pseudodeterministic algorithm that runs in time $T(|x|) \leq T(n)$, we get that $L \in \mathsf{BPTIME}[T(n)]$.
		
		Next, we show that $L$ is hard for $\mathsf{BPTIME}[n^d]$ (infinitely often and on average).
		Let $L_0 \in \mathsf{BPTIME}[n^d]$, and let $M_0$ be a probabilistic machine for $L_0$ which runs in time at most $t(n) = n^d$. Given an input $x$ for $L_0$, we define the reduction as $R(x)= \left(\langle M_0\rangle,x,1^{t(|x|)}\right)$.

	Let $\{I_n\}_{n \geq 1}$ be a polynomial-time samplable ensemble of distributions $I_n$ supported over $\{0,1\}^n$. Moreover, let $\mathcal{D}_{n, n^{d+1}}$ be the distribution supported over $\mathsf{CAPP}_{n,n^{d + 1}}$ obtained by first sampling $x \sim I_n$, then outputting the description of the circuit $C_{(M_0,x)}$. Note that $\mathcal{D}_{n, n^{d+1}}$ is also polynomial-time samplable. Therefore, algorithm $A$ succeeds with probability at least $1 - 1/n^c$ with respect to $\mathcal{D}_{n, n^{d+1}}$ on infinitely many values of $n$. For any such $n$, it follows from the definition of $R(x)$, $L$, and $\mathcal{D}_{n, n^{d+1}}$ that
	$$
		\Prob_{x\sim I_n}[L_0(x) = L(R(x))] \geq 1 - 1/n^c.
	$$
	This completes the proof.
	\end{proof}	
	
	Similarly, we note that from almost-everywhere worst-case pseudo-derandomisations of $\mathsf{CAPP}_n$ (Definition \ref{d:CAPP}), we get a $\mathsf{BPTIME}$-hard language (Definition \ref{d:bptime_hardness}).
    \begin{theorem}\label{thm:aePD-BPP_hard}
        Let $T$ be a constructive time bound.
        Suppose that there is a pseudodeterministic algorithm for $\mathsf{CAPP}_{n}$ that runs in time $T(n)$. Then there is a language in $\mathsf{BPTIME}[T(n)]$ that is $\mathsf{BPTIME}$-hard. In particular, if $T$ is a polynomial, then we have a $\mathsf{BPP}$-complete problem.
    \end{theorem}
	\begin{proof}[Proof Sketch.]
		The idea of the proof is similar to that of \Cref{thm:HeurBPP_hard}. Using an (almost-everywhere, worst-case) pseudodeterministic algorithm $A$ for solving $\mathsf{CAPP}_{n}$ as in the assumption, we can define the language $L$ as follows. Given an input $w=\left(\langle M\rangle,x,1^t\right)$ of length $n$, $w\in L$ if and only the canonical output of $A$ running on the circuit $C_{(M,x)}(y)$ of size $\leq n \cdot (\log n)^C$ that computes according to $M(x, y)$ is at least $1/2$. Since with probability at least $2/3$, $A$ outputs a \emph{fixed} good estimate of the acceptance probability of $C_{(M,x)}$, $L$ can be decided in $\mathsf{BPTIME}[T(n)]$.
	\end{proof}
    
 	Also, if $\mathsf{BPE}$ is not contained infinitely often in $\mathsf{SIZE}(2^{\varepsilon n})$ for some $\varepsilon > 0$, then we get pseudodeterministic PRGs with logarithmic seed length (see \citep{DBLP:conf/stoc/OliveiraS17}) computable in polynomial time, which can be used to pseudodeterministically approximate acceptance probabilities  of circuits in polynomial time. This leads to the following new connection between circuit lower bounds for $\mathsf{BPE}$ and the existence of complete problems for $\mathsf{BPP}$.
    
    \begin{theorem}\label{thm:LB-BPP_hard}
    If there is a language in $\mathsf{BPE}$ that is not infinitely often in $\mathsf{SIZE}(2^{\varepsilon n})$ for some $\varepsilon > 0$, then there is a $\mathsf{BPP}$-complete problem.
    \end{theorem}

	\section{An equivalence between pseudodeterminism and hierarchies}
	
	In this section, we investigate the existence of \emph{equivalences} between pseudo-derandomisations and probabilistic time hierarchies, and provide a proof of Theorem \ref{t:equivalences}. Our main result here is that a certain explicit construction problem  is ``universal'' in the following sense: it can be pseudo-derandomised if and only if a strong hierarchy theorem holds. 
	
	\subsection{Constructing strings of large \texorpdfstring{$\rKt$}{rKt} complexity versus time hierarchies}\label{s:rKt_ae}
	
	It it easy to see that a string of linear $\Kt$ complexity can be deterministically computed in exponential time. We consider the following randomised variant of this fact.
	
	\begin{hypothesis}[Pseudodeterministic construction of strings of large $\rKt$ complexity] \label{h:pseudo_constr}
		Let $T$ be a monotone constructive function with $T(\ell) \geq \ell$. There is a constant $\varepsilon > 0$ and a randomised algorithm $A$ that, given $m$, runs in time at most $T\!\left(2^{m}\right)$ and outputs with probability at least $2/3$ a fixed $m$-bit string $w_m$ such that $\rKt(w_m) \geq \varepsilon m$.\footnote{We write $T(2^m)$ instead of $T(m)$ for convenience when stating some results below. Note that this explicit construction problem cannot be solved in probabilistic time $2^{o(m)}$ by the very definition of $\rKt$.}
	\end{hypothesis} 
	
	This hypothesis can be shown to hold with $T(\ell) = \mathsf{poly}(\ell)$ under a derandomisation assumption, since in this case we get that $\Kt(x) = \Theta(\rKt(x))$ via a result from \citep{DBLP:conf/icalp/Oliveira19}.

	An algorithm for this construction problem readily implies a hierarchy theorem, as proved next.
	
	For a language $L \subseteq \{0,1\}^*$, we use $L^{=n}$ to denote $L \cap \{0,1\}^n$. We also view $L^{=n}$ as a string $\mathsf{string}(L^{=n}) \in \{0,1\}^{2^n}$, where $\mathsf{string}(L^{=n})(i) = 1$ if and only if the $i$th $n$-bit string is in $L^{=n}$. If $w$ is a $d$-bit string and $1 \leq \ell \leq d$, we let $w_{[\ell]}$ denote the $\ell$-bit string corresponding to the leftmost $\ell$ bits of $w$. 
	
	We start with the following observation, which is proved in the natural way.
	
	\begin{fact}\label{f:rKt_ub} There is a positive constant $C'$ for which the following holds.
		Let $L \in \mathsf{BPTIME}[a(n)]/b(n)$. Then for every $n \geq 1$ and $1 \leq \ell \leq 2^n$, if $v = \mathsf{string}(L^{=n})$ then 
		$$
		\rKt(v_{[\ell]}) \leq C' \cdot (\log(\ell) + \log (a(n)) + b(n) + \log (n)) + O(1).
		$$
		The same argument shows that if $L \in \mathsf{i.o.BPTIME}[a(n)]/b(n)$, then the $\rKt$ upper bound  holds for infinitely many choices of $n$ and every corresponding $1 \leq \ell \leq 2^n$.
	\end{fact}
	
	\begin{theorem}[Hypothesis \ref{h:pseudo_constr} $\Longrightarrow$ Hierarchy Theorem for Probabilistic Time] \label{t:hyp_vs_hier}
		Assume that Hypothesis \ref{h:pseudo_constr} is true for every large enough $m$. Then there are constants $k \geq 1$ and $\lambda > 0$ for which the following holds. For any constructive function $n \leq t(n) \leq 2^{\lambda \cdot 2^n}$, there is a language $L \in \mathsf{BPTIME}\left[T\!\left(t(n)^k\right)\right]$ such that $L \notin \mathsf{i.o.BPTIME}[t(n)]/\log (t(n))$. 
	\end{theorem}
	
	\begin{proof}
		Let $m(n) = \left\lceil \frac{10C'}{\varepsilon} \cdot \log t(n) \right\rceil$, where $\varepsilon$ is the constant from Hypothesis \ref{h:pseudo_constr}, and $C'$ is the constant from Fact \ref{f:rKt_ub}. Moreover, let $w_m \in \{0,1\}^m$ be the corresponding string with $\rKt(w_m) \geq \varepsilon m$. Define the following language $L$. On inputs of length $n$, $\mathsf{string}(L^{=n})(i) = 0$ if $i > m(n)$, and $\mathsf{string}(L^{=n})(i) = w_m(i)$ otherwise. Note that this is well defined, since by an appropriate choice of $\lambda$ in the upper bound for $t(n)$ we get $m(n) \leq 2^n$.  
		
		By construction, we have that $L \in \mathsf{BPTIME}\left[T\!\left(2^{m(n)}\right)\right]$, which places $L \in \mathsf{BPTIME}\left[T\!\left(t(n)^k\right)\right]$ for a fixed $k \geq 1$ that is independent of $t(n)$. On the other hand, if we let $\ell = m(n) \leq 2^n$, it is not hard to see via Fact \ref{f:rKt_ub} (using our choice of $m(n)$ when computing $L$ on inputs of length $n$) that $L \notin \mathsf{i.o.BPTIME}[t(n)]/\log (t(n))$.
	\end{proof}

	It is not hard to see that \Cref{t:hyp_vs_hier} is in fact  equivalent to Hypothesis \ref{h:pseudo_constr} when $T(\ell) = \mathsf{poly}(\ell)$. This is obtained by viewing the hard language $L \in \mathsf{BPTIME}[t(n)^k]$ for the maximum admissible $t(n)$ in Proposition \ref{t:hyp_vs_hier} as a sequence of strings of length $m = 2^n$ that can be pseudodeterministically constructed in time $2^{O(m)}$. 
	
	\begin{theorem}[Hierarchy Theorem for Probabilistic Time $\Longrightarrow$ \Cref{h:pseudo_constr}]\label{t:hierarchy-to-pseudo_constr}
		Let $T(\ell)\geq \ell$ be a monotone constructive time bound.	Suppose there are constants $k \geq 1$ and $\lambda > 0$ for which the following holds: for any constructive function $n \leq t(n) \leq 2^{\lambda \cdot 2^n}$, there is a language $L \in \mathsf{BPTIME}\left[T\!\left(t(n)^k\right)\right]$ such that $L \notin \mathsf{i.o.BPTIME}[t(n)]/\log (t(n))$. Then Hypothesis \ref{h:pseudo_constr} is true. 
	\end{theorem}
	\begin{proof}
		Given $m$, we show how to pseudodeterministically output an $m$-bit string with $\rKt$ complexity  $\Omega(m)$. Let $n = \lfloor\log (m)\rfloor$ and $t= t(n) = 2^{\lambda\cdot 2^{n}/c} \leq 2^{\lambda\cdot m/c}$, where $c>0$ is some sufficiently large constant. Then we output the string $y$, where
		\[
		y \eqdef \mathsf{string}(L^{=n})\circ 0 ^{m-2^n}.
		\] 
		It is clear that $y$ can be output with high probability in time $2^n\cdot T\!\left(t(n)^k\right) \cdot \poly(n) \leq T\!\left(2^{m}\right)$, where the $\poly(n)$ factor accounts for error reduction and we use the fact that $T(\ell)\geq \ell$ and $T$ is monotone.
		
		Next, we show that $\rKt(y)=\Omega(m)$. It suffices to show that $\rKt(\mathsf{string}(L^{=n}))=\Omega(m)$. For the sake of contradiction, suppose $\rKt(\mathsf{string}(L^{=n}))=o(m)$. Then there is some advice string $\alpha$ of $o(m)=o(\log t)$ bits such that the universal probabilistic Turing machine takes $\alpha$ as input, runs in time $2^{o(m)}=t^{o(1)}$ and outputs $\mathsf{string}(L^{=n})$. This contradicts our assumption that $L \notin \mathsf{i.o.BPTIME}[t]/\log (t)$.
	\end{proof}

		\subsection{Hierarchies from weaker pseudodeterministic explicit constructions}\label{s:hier_from_expl_constr}
	
	In this section, we consider a  variant of Hypothesis \ref{h:pseudo_constr} and how it relates to existing results and techniques. 
	
	\begin{definition}[$R_{n,d}$]\label{d:rKt-problem}
		For an integer $d>0$, we define $R_{n,d}$ to be the search problem of given $1^n$ outputting a string $y$ of $\lceil d\cdot \log(n) \rceil$ bits such that $\rKt(y)\geq \lceil d\cdot \log(n)\rceil/2$.
	\end{definition}
	
	As opposed to the presentation in Section \ref{s:rKt_ae}, here we consider weak pseudo-deterministic algorithms for solving $R_{n,d}$ that might not succeed on every input length. In a bit more detail, by a pseudodeterministic algorithm for $R_{n,d}$ that succeeds infinitely often, we mean that the algorithm maintains a pseudo-deterministic behaviour on every input string, but is only guaranteed to output a string of large $\rKt$ complexity for infinitely many input lengths.
	
	\begin{theorem}[Pseudodeterministic constructions for {$R_{n,d}$} yield probabilistic time hierarchies]\label{t:pd-rKt-to-hierarchy}
		For every constant $k\geq 1$ there is a constant $d\geq 1$ for which the following holds. If there is a pseudodeterministic algorithm for $R_{n,d}$ that runs in time $T(n^d)$ and succeeds for infinitely many values of $n$, then there is a language $L \in \mathsf{BPTIME}[T(n^d)]$ such that $L \notin \mathsf{BPTIME}[n^k]/(k\cdot\log(n))$. 
	\end{theorem}
	\begin{proof}	
	The argument is analogous to the proof of Theorem \ref{t:hyp_vs_hier}, and we omit the details.
		\end{proof}

  Note that a trivial (pseudo)deterministic algorithm for $R_{n,d}$ would run in time roughly $2^{n^{d/2}}$, since strings of $\rKt$ complexity $d \cdot \log(n)/2$ refer to probabilistic algorithms running in time $n^{d/2}$, and estimating their acceptance probability in a trivial way would take time of order $2^{n^{d/2}}$. Next, we adapt existing techniques to obtain an unconditional sub-exponential time algorithm for this explicit construction problem. 
	
	\begin{theorem}[Sub-exponential time pseudodeterministic construction for {$R_{n,d}$}]\label{t:uncondition-pd-rkt}
		For every constant $\varepsilon>0$ and positive integer $d$, there is a pseudodeterministic algorithm for $R_{n,d}$ that runs in time $2^{n^{\varepsilon}}$ and succeeds for infinitely many values of $n$.
	\end{theorem}

	\begin{proof}
		We first consider a probabilistic algorithm $B$ such that, on input $a\in\bool^{m}$ where $m=\lceil d \cdot \log(n) \rceil$, $B$ rejects with probability $\geq 2/3$ if $\rKt(a)<m/2$ and accepts with probability $\geq 2/3$ if $\rKt(a)\geq 3m/4$. It was shown in \cite{DBLP:conf/icalp/Oliveira19} that $B$ can be made to run in time $2^{O(m)}=n^{O(d)}$. For $a\in\bool^{\lceil d\cdot\log(n) \rceil}$, let $C_n^{a}$ be the Boolean circuit such that on input $y\in\bool^{n^{O(d)}}$, $C_n^{a}(y)$ is $1$ if and only if $B$ accepts $a$ using $y$ as its randomness. Note that each $C_n^{a}$ has size at most $s=n^{O(d)}$. Also, let $A$ be the (i.o.-)pseudodeterministic algorithm for $\mathsf{CAPP}_{n,s}$ in \Cref{t:uncondition-pd-CAPP} that runs in time $2^{n^{\varepsilon/2}}$ and succeeds with probability at least $1-1/(3n^d)$ over any polynomial-time samplable distribution. We assume without loss of generality, using amplification if necessary, that $A$ outputs the canonical answer with probability at least $1-1/(100 n^{d})$.
		\begin{algorithm}
			\caption{Infinitely often pseudodeterministic construction of strings of large $\rKt$ complexity}
			\begin{algorithmic}[1]
				
				\Procedure{$D$}{$1^n, d$} 
				
				\For{$a\in \bool^{\lceil d\cdot\log(n) \rceil}$}	
				\State $\mu_{a} = A(1^n, C_n^{a})$
				
				\If{$\mu_{a} > 1/3 + 1/10$}
				\State output $a$
				\EndIf
				
				\EndFor
				
				\State Output ``Fail''
				\EndProcedure
			\end{algorithmic}
		\end{algorithm}
		
		We now argue the correctness of the above algorithm. 
		Note that for every $n$, by a union bound over $a\in\bool^{\lceil d\cdot\log(n) \rceil}$, $A(1^n, C_n^{a})$ outputs the canonical $\mu_{a}$ for every $a$ with high probability, in which case the final output of the algorithms is fixed, so $A$ is pseudodeterministic. 
		
		Now consider the polynomial-time samplable distribution $\mathcal{D}_{n, s}$ supported over $\mathsf{CAPP}_{n,s}$ obtained by first sampling $a \sim \bool^{\lceil d\cdot\log(n) \rceil}$, then outputting the description of the circuit $C_n^{a}$. Note that each $C_n^{a}$ has probability weight at least $1/(2n^d)$. Since $A$ succeeds with probability at least $1-1/(2n^d)$ over $\mathcal{D}_{n, s}$, we conclude that $A$ succeeds on \emph{every} $C_n^{a}$, for infinitely many values of $n$. For any such $n$, $\mu_{a}$ is a good estimate of $\Pr_y[C_n^{a}(y) = 1]$, for every $a\in\bool^{\lceil d\cdot\log(n) \rceil}$. Then by the definition of $C_n^{a}$, an output $a$ of the algorithm cannot have $\rKt$ less than $m/2$ since the algorithm $B$ accepts $a$ with probability less than $1/3$ and $\mu_{a}$ should be less than $1/3+1/10$. Also, note that since we enumerate every $a$ in $\bool^{\lceil d\cdot\log(n) \rceil}$, $B$ must accept at least one $a$, and in this case we have $\mu_{a}\geq 2/3-1/10\geq 1/3+1/10$. (Note that the algorithm may output a string outside of $B$'s YES promise, but such a string will also have $\rKt$ complexity at least $m/2$, and this output is fixed as long as $A$ gives the canonical $\mu_{a}$ for every $a$, which happens with high probability.)
	\end{proof}
	
	As an immediate consequence of Theorems \ref{t:pd-rKt-to-hierarchy} and \ref{t:uncondition-pd-rkt}, we can recover a known hierarchy theorem for probabilistic time, which says that there is a language $L \in \mathsf{BPTIME}\left[2^{n^{\varepsilon}}\right] \setminus \mathsf{BPTIME}\left[n^k\right]$. Furthermore, if \Cref{t:uncondition-pd-rkt} could be improved either with a better running time or with a pseudo-deterministic simulation that works on every large enough input length, new hierarchies results for probabilistic time would follow.

	\subsection{On the pseudo-derandomisation of \texorpdfstring{$\mathsf{unary}$-$\mathsf{BPE}$-$\mathsf{search}$}{unary-BPE-Search}} \label{sec:unary_BPE}
	
	Consider the following hypothesis about the pseudo-derandomisation of $\mathsf{unary}$-$\mathsf{BPE}$-$\mathsf{search}$.  
	
	\begin{hypothesis}[Pseudo-derandomisation of  $\mathsf{unary}$-$\mathsf{BPE}$-$\mathsf{search}$]\label{h:pd-BPE-search}
		For every $\mathsf{unary}$-$\mathsf{BPE}$-$\mathsf{search}$ relation $R$, there is a pseudodeterministic search algorithm for $R$ that runs in exponential time.
		In other words, there is a pair $(A,B)$ of  probabilistic algorithms witnessing that $R \in \mathsf{unary}$-$\mathsf{BPE}$-$\mathsf{search}$, where $A$ and $B$ run in time exponential in $n$, and on every input $x = 1^n$ there is a string $y$ such that $\Pr_A[A(1^n) = y] \geq 2/3$.
	\end{hypothesis}
	
	First, we observe that an \emph{average-case} pseudo-derandomisation of $\mathsf{BPP}$-$\mathsf{search}$ leads to a \emph{worst-case} pseudo-derandomisation of $\mathsf{BPE}$-$\mathsf{search}$.

	\begin{proposition}[Pseudo-derandomisation of $\mathsf{BPP}$-$\mathsf{search}$ on average $\Longrightarrow$ Pseudo-derandomisation of $\mathsf{BPE}$-$\mathsf{search}$]\label{p:pd-BPP-topd-BPE}
		Let $T$ be a constructive time bound.
		Suppose that for every $\mathsf{BPP}$-$\mathsf{search}$ problem $R$ and for every polynomial-time samplable ensemble $\{\mathcal{D}_n\}_{n \geq 1}$, there is a pseudodeterministic algorithm $\mathcal{A}$ for $R$ that runs in time $T(n)$ and succeeds with probability at least $1-1/(3n)$ over inputs from $\mathcal{D}_n$. Then there is a pseudodeterministic search algorithm for each relation in $\mathsf{BPE}$-$\mathsf{search}$ that runs in time $T\!\left(2^n\right)$.
	\end{proposition}
	\begin{proof}
		Let $R_0$ be a $\mathsf{BPE}$-$\mathsf{search}$ problem with a search algorithm $A_0$ and a verification algorithm $B_0$. We show how to solve $R_0$ assuming the pseudo-derandomisation of $\mathsf{BPP}$-$\mathsf{search}$. Consider the following search problem $R$. For a pair $(x,y)$ where $x\in\bool^n$ and $y\in\bool^*$, $(x,y)\in R$ if and only if $x$ is of the form $1^{n - \lceil \log n \rceil}i$ for some $i \in \{0,1\}^{\lceil \log n \rceil}$ and $(i,y)\in R_0$. Note that $R$ is a $\mathsf{BPE}$-$\mathsf{search}$ problem: its search algorithm can be defined as $A(x)=A_0(i)$, and its verification algorithm $B$ first checks if $x$ has the correct form then invokes $B_0(i,y)$.
		
		For an integer $n$, let $\mathcal{D}_n$ be the polynomial-time samplable distribution which samples a random  string of length $\lceil \log n \rceil$ and appends it to the string $1^{n-\lceil \log n \rceil}$. Let $C$ be a pseudodeterministic algorithm that runs in time $T(n)$ and solves $R$ with probability at least $1-1/(3n)$ over inputs from the distribution $\mathcal{D}_n$.
		
		To solve the search problem $R_0$ on an given input $i\in\bool^{m}$, we first construct the input of $x=1^{2^{m}-m}i$ and then output $C(x)$. It is easy to see that if $C$ pseudodeterministically solves the problem $R$ on $x$, then the above approach pseudodeterministically solves $R_0$ on $i$ in time $T\left(2^m\right)$. However, we only have that $C$ succeeds with probability at least $1-1/\left(3\cdot 2^m\right)$ over $\mathcal{D}_{2^m}$. But note that  $\mathcal{D}_{2^m}$ is uniform over the set $S=\left\{1^{2^{m}-m}i\right\}_{i\in\bool^m}$, where $|S|=2^m$. This means that $C$ succeeds on every input in $S$, and hence the above approach pseudodeterministically solves $R_0$ on every input.
	\end{proof}
	
	\begin{proposition}\label{p:pd-BPE-search-to-pseudo_constr}
		\Cref{h:pd-BPE-search} $\Longrightarrow$ \Cref{h:pseudo_constr} with $T(\ell) = \mathsf{poly}(\ell)$.
	\end{proposition}
	\begin{proof}
		Let $R$ be the following relation:
		\[
		(1^m, y) \in R \iff |y|=m \text{ and } \rKt(y)\geq 0.1 m.
		\]
		To prove the proposition, it suffices to show that the (total) unary relation $R \in  \mathsf{unary}$-$\mathsf{BPE}$-$\mathsf{search}$. Consider the following search algorithm $A$ that, on input $1^m$, outputs a string in $\bool^m$ uniformly at random. By a counting argument, with probability at least $2/3$, the string output by $A$ has $\rKt$ complexity at least $0.2m$, which satisfies the condition of $R$. Let $B$ be a probabilistic algorithm that solves the $\mathsf{Gap}$-$\mathsf{MrKtP}$ problem, i.e., it rejects (in the sense of a bounded-error probabilistic algorithm) strings with $\rKt$ complexity less than $0.1m$ and accepts strings with $\rKt$ complexity at least $0.2m$. It was shown in \cite{DBLP:conf/icalp/Oliveira19} that $B$ can be made to run in time $2^{O(m)}$. Therefore, $B$ is our verification algorithm that rejects the negative instances of $R$ and accepts at least a $2/3$-fraction of $A$'s outputs.
	\end{proof}
	
	We leave open the following question.
	
	\begin{question}
		Is it the case that \Cref{h:pseudo_constr} with $T(\ell) = \mathsf{poly}(\ell)$ implies \Cref{h:pd-BPE-search}?  
	\end{question}
	
	A positive solution would establish the equivalence between strong probabilistic time hierarchies, the explicit construction problem for $\rKt$, and the pseudo-derandomisation of unary $\mathsf{BPE}$-$\mathsf{search}$.

	\section*{Acknowledgements}
	
	We thank Peter Dixon, A. Pavan and N. V. Vinodchandran for bringing their independent unpublished work \cite{DPV21b} to our attention. We are also grateful to Lijie Chen for sharing comments about a preliminary version of the paper that helped us to  improve the presentation.
	
	The first two authors received support from the Royal Society University Research Fellowship URF$\setminus$R1$\setminus$191059.

	\bibliographystyle{alpha}	
	
	\bibliography{main}	
	
	\appendix


	\section{On the pseudo-derandomisation of \texorpdfstring{$\mathsf{CAPP}$}{CAPP} from \texorpdfstring{\citep{DBLP:conf/stoc/OliveiraS17}}{[OS17]}}\label{s:appendix_OS17}
		
		In this section, we verify that the proof of an unconditional (average case, infinitely often, sub-exponential time) pseudo-derandomisation of $\mathsf{CAPP}$ from \citep{DBLP:conf/stoc/OliveiraS17} guarantees a pseudo-deterministic output on \emph{every} input string.
	
	\begin{theorem}[Reminder of \Cref{t:uncondition-pd-CAPP}]\label{t:uncondition-pd-CAPP-app}
		For any constants $\varepsilon>0$ and $c, d\geq 1$, there is a pseudodeterministic algorithm for $\mathsf{CAPP}_{n,n^{d}}$ that runs in time $2^{O(n^{\varepsilon})}$, and for any polynomial-time samplable ensemble of distributions $\mathcal{D}_{n, n^d}$ supported over circuits of size $\leq n^d$, succeeds with probability $1-1/n^c$ over $\mathcal{D}_{n, n^d}$ for infinitely many values of $n$. 
	\end{theorem}
	\begin{proof}[Sketch of the proof]
		We follow the analysis from \citep{DBLP:conf/stoc/OliveiraS17} and consider two cases.
		
		Suppose that $\mathsf{PSPACE}\subseteq \mathsf{BPP}$. First consider the problem of given a circuit $C$ of length $n^{d}$ and $j\in [n^d]$, output the $j$-bit of $\beta$, the number of satisfying assignments of $C$. Note that this problem can be computed using $n^{O(d)}$ space by enumerating all possible inputs for $C$. By our assumption, this problem can also be solved in randomized time $n^{O(d)}$. Therefore, we have a $n^{O(d)}$ time randomized algorithm to compute exactly the acceptance probability of $C$, and we are done.

		Now assume  $\mathsf{PSPACE}\not\subseteq \mathsf{BPP}$. Suppose that we are given a circuit $C$ with $|C| = n^{d}$. Consider \Cref{t:io-PRG} with $b=d/\varepsilon$ and the generator $G_{\ell}$, where $\ell= \lceil n^{\varepsilon}\rceil$. We then output
		\[
		\mu\eqdef\Prob_{z\in \bool^{\ell}}[C(G_{\ell}(z))=1].
		\]
		It is easy to see that the running time is $2^{O(n^{\varepsilon})}\cdot n^{O(d)}=2^{O(n^{\varepsilon})}$.
		
		Arguing in a slightly informal way for simplicity (with respect to uniformity and samplability), let $n$ be such that for $\ell= \lceil n^{ \varepsilon}\rceil$, $G_\ell$ is a generator whose output cannot be distinguished from random on average by polynomial-time samplable circuits, assuming $\mathsf{PSPACE}\not\subseteq \mathsf{BPP}$ (since the function mapping $n$ to $\lceil n^{ \varepsilon}\rceil$ is surjective, this happens infinitely often). Then, for any distribution $\mathcal{D}_{n, n^d}$ samplable in time $\poly(n^d)$ and any constant $c$, with probability at least $1-1/n^c$ over $C \sim \mathcal{D}_{n^d}$ we have
		\begin{equation}\label{e:reduction-fail}
			\left| \Prob_{y\in\bool^{n^{d}}}[C(y)=1] -\mu \right| \leq 1/10,
		\end{equation}
		where this inequality relies on the security of $G_{\ell}$.
		
		Note that in both cases the resulting algorithm is pseudo-deterministic on \emph{every} input string. This is because in the first case (i.e.~when $\mathsf{PSPACE} \subseteq \mathsf{BPP}$) the algorithm is correct and pseudo-deterministic on every input string. In the other case, while the algorithm might fail on some inputs, it is a \emph{deterministic} algorithm (since the PRG from \Cref{t:io-PRG} is computed by a deterministic algorithm). 
	\end{proof}

	
	\section{Pseudo-derandomisations for \texorpdfstring{$\mathsf{BPP}$-$\mathsf{search}$}{BPP-search} and their consequences} \label{s:BPP_search_appendix}
	
	In this section, we establish connections between weak pseudo-derandomisations of $\mathsf{BPP}$-$\mathsf{search}$ and structural results for probabilistic time.\footnote{We note that \citep{DBLP:journals/corr/Holden17} claims an unconditional pseudo-derandomisation of $\mathsf{BPP}$-$\mathsf{search}$. However, their argument seems to require a stronger condition on the verifier machine $B$, namely, that on every input pair $(x,y)$ the probability that $B$ accepts $(x,y)$ is bounded away from $1/2$. This appears to be necessary in the pseudo-derandomisation argument from \citep{DBLP:journals/corr/Holden17} to maintain a pseudodeterministic output when computing the first solution accepted by $B$.}
	
	First, we obtain hierarchies from weak pseudo-derandomisations of $\mathsf{BPP}$-$\mathsf{search}$.
	
		\begin{proposition}[i.o.-pseudo-derandomisation of $\mathsf{BPP}$-$\mathsf{search}$ over samplable distributions $\Longrightarrow$ probabilistic time hierarchy theorem]\label{p:io-pd-to-hierarchy}
		Let $T$ be a time-constructible function. Suppose that for every $\mathsf{BPP}$-$\mathsf{search}$ problem $R$ and for every polynomial-time samplable ensemble $\mathcal{D}_n$, there is a pseudodeterministic search algorithm $A$ for $R$ that runs in time $T(n)$, and for infinitely many input lengths $n$, succeeds with probability at least $1-1/(3n)$ over inputs from $\mathcal{D}_n$.\footnote{In other words, for every input $x$ there is a canonical output $y$ for $x$ such that $\Pr_A[A(x) = y] \geq 2/3$, and on infinitely many values of $n$, except with probability at most $1/(3n)$ over $x \sim \mathcal{D}_n$, we have that $y \in R_x$ and $\Pr_B[B(x,y) = 1] \geq 2/3$, where $B$ is the verification algorithm associated with $A$.}
		
		Then, for every $k \geq 1$ there is a language $L \in \mathsf{BPTIME}[T(n)] \setminus \mathsf{BPTIME}\left[n^k\right]$. Moreover, if the pseudodeterministic simulation succeeds with probability at least $1 - 1/(3n)$ on every large enough input length $n$, then $L \in \mathsf{BPTIME}[T(n)] \setminus \mathsf{i.o.BPTIME}\!\left[n^k\right]$.
	\end{proposition}

	\begin{proof}
		Let $B_1, B_2, \ldots$ be an enumeration of all probabilistic machines, and consider the following relation $R$. For a pair $(x,y)$ where $x\in\bool^n$ and $y\in[0,1]$ in represented as a binary string, $(x,y)\in R$ if and only if $x$ is of the form $x$ form $1^{n - \lceil \log n \rceil}i$ for some $i \in \{0,1\}^{\lceil \log n \rceil}$ and
		\[
		|y-\mu| \leq 0.1,
		\]
		where $\mu$ is the probability that the $i$-th probabilistic machine $B_i$ accepts $x$ when running for $n^k$ steps.
		
		We first show that $R\in \mathsf{BPP}$-$\mathsf{search}$. The search algorithm $A$, on input $x=1^{n - \lceil \log n \rceil}i$, (repeatedly) simulates $B_i$ on $x$ for $n^k$ steps and with probability at least $2/3$, outputs a value $\alpha$ that is at most $0.03$ far from the acceptance probability of $B_i$ on $x$. It is clear that $A$ can be made to run in probabilistic polynomial time and that it outputs a value that satisfies the condition of $R$ with probability at least $2/3$ (via a standard concentration bound). The verification algorithm $B$, will first check if $x$ has the correct form, and then invoke a probabilistic algorithm $B_0$ such that with probability at least $2/3$, $B_0$ outputs a value $\beta$ that is at most $0.03$ far from the acceptance probability of $B_i$ on $x$ (again using a standard argument and a concentration bound). $B$ accepts iff $|\beta-y| \leq 0.06$.  On the one hand, $B$ rejects all the bad $y$'s (those that are $> 0.1$ far from the correct acceptance probability) with probability at least 2/3 (when $B_0$ outputs a value $\beta$ that is at most 0.03 far and hence $|\beta-y| > 0.07$); on the other hand, with probability at least 2/3 (over the randomness of $A$), $A$ outputs a value that is at most $0.03$ far, in which case $B$ accepts this output of $A$ with probability at least $2/3$ (again when $B_0$ outputs a value that is at most 0.03 far).
		
		Next, we define the (hard) language $L$. Let $A'$ be a (i.o.-)pseudodeterministic search algorithm for $R$ (that succeeds  with high probability over a particular polynomial-time samplable input distribution defined below). Let $L$ be as follows: 
		\[
		x \in L \iff \text{the canonical output of $A'(x)$ has a value that is less than $1/2$}.
		\]
		
		Since $A'$ is pseudodeterministic, it is easy to see that $L \in \mathsf{BPTIME}[T(n)]$.
		Next, we show that $L\not\in\mathsf{BPTIME}\!\left[n^k\right]$. Let $L'$ be an arbitrary language in $\mathsf{BPTIME}\!\left[n^k\right]$. Then there is an $i$ such that the machine $B_i$ computes $L'$ and always stops in at most $n^k$ steps. Let $n\geq i$ be such that our pseudodeterministic algorithm $A'$ succeeds on inputs of length $n$ coming from the distribution $\mathcal{D}_n$ defined by sampling a random string of length $\log n$ and appending it to the string $1^{n - \log n}$ (note that our choice for the ensemble of distributions is independent of the other parameters). Assume without loss of generality that $1^{n - \lceil \log n \rceil}i\in L'$. Then we have $\mu=\Prob\left[B_i\left(1^{n - \lceil \log n \rceil}i\right)=1\right]\geq 2/3$. Note that distribution $\mathcal{D}_n$ is uniform over a set of size at most $2n$. Since our pseudodeterministic algorithm $A'$ succeeds with probability at least $1-1/(3n)$ over such an input distribution, we have that $A'$ succeeds on every input in its support, including $1^{n - \lceil \log n \rceil}i$. In other words, the canonical output of $A'$ on $1^{n - \lceil \log n \rceil}i$ is at least $2/3-0.1>1/2$, which implies that $1^{n - \lceil \log n \rceil}i \not\in L$. This shows that $L \neq L'$, and since $L' \in \mathsf{BPTIME}\!\left[n^k\right]$ was arbitrary, it follows that $L \notin \mathsf{BPTIME}\!\left[n^k\right]$.
		
		It is easy to check that the ``moreover'' statement follows from a similar argument.
	\end{proof}
	
	Next, we show how to get completeness results from strong pseudo-derandomisations of $\mathsf{BPP}$-$\mathsf{search}$. Consider the following hypothesis.

	\begin{hypothesis}[Statement $H(T)$]\label{h:pd-BPP-search} Let $T$ be a time-constructible function.
		For every $\mathsf{BPP}$-$\mathsf{search}$ problem $R$, there is pseudodeterministic search algorithm for $R$ that runs in time $T$.
		More precisely, there is a pair $(A,B)$ of  probabilistic algorithms witnessing that $R \in \mathsf{BPP}$-$\mathsf{search}$, where $B$ runs in time polynomial in $|x|$, $A$ runs in time $T(|x|)$, and for every input $x$ there is a string $y$ such that $\Pr_A[A(x) = y] \geq 2/3$.
	\end{hypothesis}

	\begin{theorem}[Pseudo-derandomisation of $\mathsf{BPP}$-$\mathsf{search}$ yields $\mathsf{BPTIME}$-hard problems]\label{t:pd-imply-BPP-hard-problem}
		If \Cref{h:pd-BPP-search} holds for a time-constructible $T$, then there exists a $\mathsf{BPTIME}$-hard problem in $\mathsf{BPTIME}[O(T(n))]$.
	\end{theorem}
	
	\begin{proof}
		Consider the following relation $R$:
		\[
		\left\{\left(\langle M\rangle,x,1^t\right),\mu\right\} \in R \iff \mu =\Prob[\text{$M$ accepts $x$ in $\leq t$ steps}] \pm 0.1.
		\]
		We claim that $R \in  \mathsf{BPP}$-$\mathsf{search}$. First, note that by (repeatedly) simulating $M$ on $x$ for at most $t$ steps, we can design a probabilistic polynomial-time search algorithm $A$, such that with probability at least $2/3$, $A$ outputs a value $\alpha$ that is at most $0.03$ far from the acceptance probability of $M$ on $x$. It is clear that $A$ outputs a value that satisfies the condition of $R$ with probability at least $2/3$. For the verification algorithm, we first use a probabilistic algorithm $B_0$ such that with probability at least $2/3$, $B_0$ outputs a value $\beta$ that is at most $0.03$ far from the acceptance probability. We then let the verification algorithm $B$ be such that, on input $\left\{\left(\langle M\rangle,x,1^t\right), \mu\right\}$, $B$ accepts iff $|\beta-\mu| \leq 0.06$. On the one hand, $B$ rejects all bad $\mu$ (those that are $> 0.1$ far from the correct acceptance probability) with probability at least 2/3 (when $B_0$ outputs a value $\beta$ that is at most 0.03 far and hence $|\beta-\mu| > 0.07$); on the other hand, with probability at least 2/3 (over the randomness of $A$), $A$ outputs a value that is at most $0.03$ far, in which case $B$ accepts this output of $A$ with probability at least $2/3$ (again when $B_0$ outputs a value that is at most 0.03 far).

		Assuming \Cref{h:pd-BPP-search}, let $C$ be a pseudodeterministic search algorithm for $R$. That is, on input $w=\left(\langle M\rangle,x,1^t\right)$, $C$ runs in time $T(|w|)$ and with probability at least $2/3$ outputs a fixed value $\mu^*$, which is a good estimate of the acceptance probability of the machine $M$ running on $x$ in $t$ steps.
		Let's define a language $L$ as follows: 
		\[
		w\in L \iff \text{the canonical output of $C(w)$ has a value that is  at least $1/2$}.
		\]
		Next, we show that $L$ is $\mathsf{BPTIME}$-hard with respect to deterministic polynomial-time reductions.
		Let $L'\in \mathsf{BPTIME}[t(n)]$, and let $M_{L'}$ be a corresponding bounded-error machine that decides $L$ under this time bound. Consider an instance $x$ for $L'$. We let the reduced instance for $L$ be $w=\left(\langle M_{L'}\rangle,x,1^{t(|x|)}\right)$. It is easy to verify that $w$ can be produced in $\poly(t(|x|))$ time deterministically, for a fixed polynomial that is independent of $t$. Let's assume that $x \in L'$ (the other case is analogous), which means $M_{L'}$ accepts $x$ with probability at least $2/3$ (within $t$ steps). In this case, our pseudodeterministic algorithm $C$ on input $w$ will output (with probability at least $2/3$) a fixed number $\mu^*$ that is a good estimate of the acceptance probability of $M_{L'}$ on $x$, which means $\mu^*$ is at least $1/2$. Hence the canonical output value of $A(x)$ is at least $1/2$. By definition, $w\in L$. This shows the $\mathsf{BPTIME}$-hardness of $L$.
		
		Finally, to see that $L$ is in $\mathsf{BPTIME}[O(T(n))]$, note that on input $w$, we can (repeatedly) run the algorithm $C$ to (confidently) find out the canonical output of $C(w)$, since $C$ is pseudodeterministic.
	\end{proof}
	
	As a consequence of the results described above, we obtain the following corollaries.
	
	\begin{corollary}[Efficient pseudo-derandomisation of $\mathsf{BPP}$-$\mathsf{search}$ implies $\mathsf{BPP}$-complete problems]\label{c:pd-imply-BPP-hard-problem}
		If for every $\mathsf{BPP}$-$\mathsf{search}$ problem $R$  there is a  pseudodeterministic polynomial-time search algorithm  for $R$, then there is a  $\mathsf{BPP}$-complete problem.
	\end{corollary}

	\begin{corollary}[Probabilistic Time Hierarchy from Pseudo-derandomisation]\label{c:hierarchy}
		If \Cref{h:pd-BPP-search} holds for a time-constructible function $T$, there is a constant $c$ such that for every time-constructible $t$,
		\[
		\mathsf{BPTIME}\left[T\!\left(t(n)^{c}\right)\right]\not\subseteq	\mathsf{BPTIME}[t(n)].
		\]
	\end{corollary}
	\begin{proof}
		This follows from Theorem \ref{t:pd-imply-BPP-hard-problem} using the argument from the proof of {\cite[Theorem 3.6]{DBLP:conf/random/Barak02}}.
	\end{proof}
	 
	 \newpage 
	 
	\section{Pseudodeterminism and the structure of probabilistic time} \label{s:summary_appendix}
	
		The diagram below summarises several connections established in this work.  We note that a similar diagram of implications also hold in the context of   pseudo-derandomisations of the Circuit Acceptance Probability Problem ($\mathsf{CAPP}$).\\~\\
		
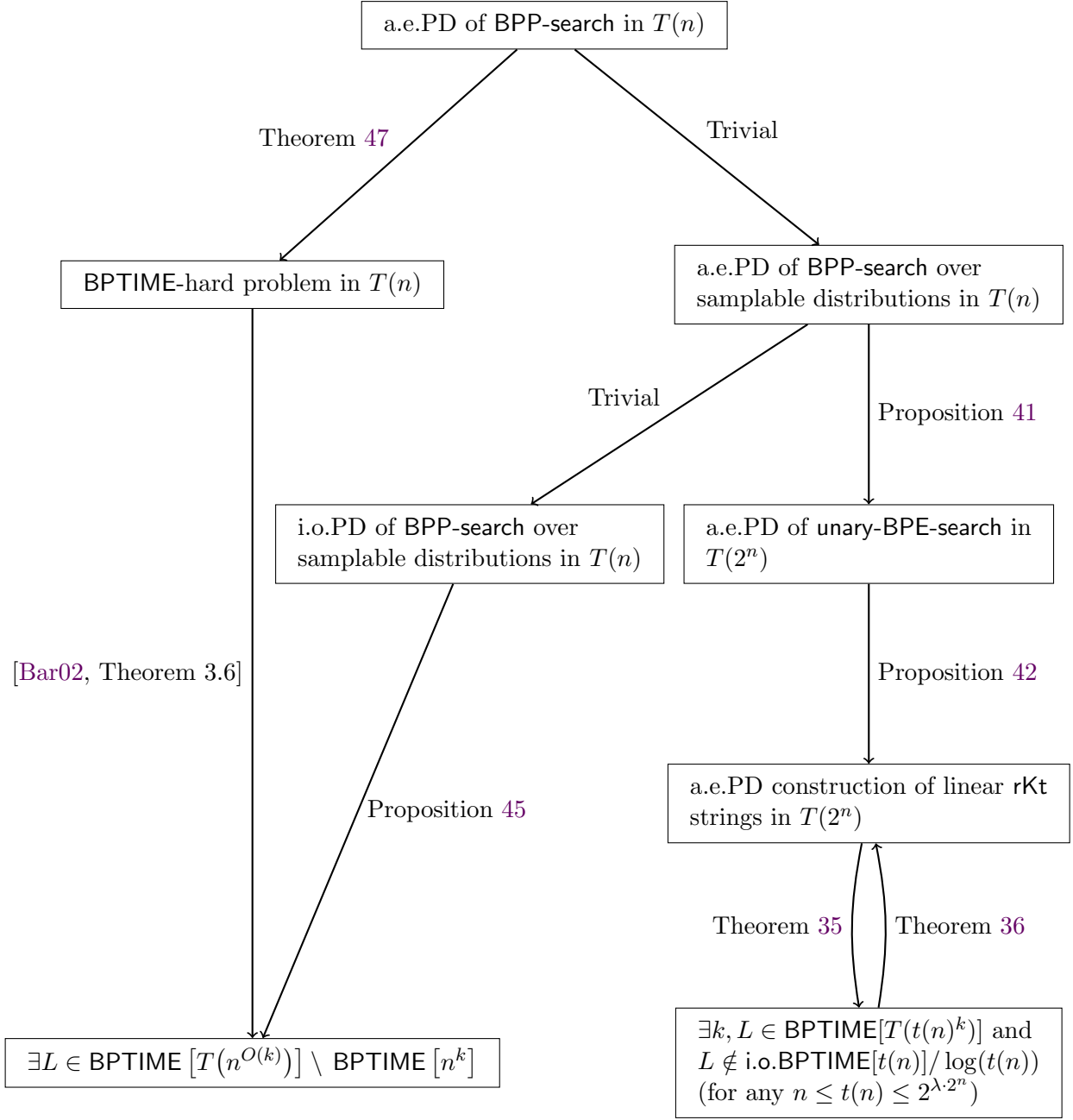
		\begin{tikzpicture}
			
			\pgfdeclarelayer{nodelayer}
			\pgfdeclarelayer{edgelayer}
			\pgfsetlayers{nodelayer,edgelayer}
			
			\begin{pgfonlayer}{nodelayer}
				\node (0) at (-5, 6) [rectangle,draw]
				{\begin{tabular}{l} a.e.PD of $\mathsf{BPP}$-$\mathsf{search}$ in $T(n)$\end{tabular}};
				
				\node (1) at (-9.5, 2) [rectangle,draw]
				{\begin{tabular}{l} $\mathsf{BPTIME}$-hard problem in $T(n)$\end{tabular}};
				
				\node (2) at (0, 2) [rectangle,draw]
				{\begin{tabular}{l} a.e.PD of $\mathsf{BPP}$-$\mathsf{search}$ over \\samplable distributions in $T(n)$ \end{tabular}};
				
				\node (7) at (-6.15,-2) [rectangle,draw]
				{\begin{tabular}{l} i.o.PD of $\mathsf{BPP}$-$\mathsf{search}$ over \\samplable distributions in $T(n)$ \end{tabular}};
				
				\node (3) at (0, -2) [rectangle,draw] 
				{\begin{tabular}{l} a.e.PD of $\mathsf{unary}$-$\mathsf{BPE}$-$\mathsf{search}$ in \\$T(2^n)$ \end{tabular}};
				
				\node (4) at (0, -6) [rectangle,draw] 
				{\begin{tabular}{l} a.e.PD construction of linear $\rKt$ \\strings in $T(2^n)$ \end{tabular}};
				
				\node (5) at (0, -10) [rectangle,draw]
				{\begin{tabular}{l} $\exists k, L \in \mathsf{BPTIME}[T(t(n)^k)]$ and \\ $L \notin \mathsf{i.o.BPTIME}[t(n)]/\log (t(n))$ \\(for any $n \leq t(n) \leq 2^{\lambda \cdot 2^n}$) \end{tabular}};
				
				\node (6) at (-9.5, -10) [rectangle,draw]
				{\begin{tabular}{l} $\exists L \in \mathsf{BPTIME}\left[T\! \left(n^{O(k)}\right)\right] \setminus\, \mathsf{BPTIME}\left[n^k\right]$\end{tabular}};`
			\end{pgfonlayer}
			\begin{pgfonlayer}{edgelayer}
				\draw [->, thick] (0) to node[above left]  {\Cref{t:pd-imply-BPP-hard-problem}} 
				(1) ;
				
				\draw [->, thick] (0) to node[above right]
				{Trivial}
				(2);
				
				\draw [->, thick] (2) to node[right] {\Cref{p:pd-BPP-topd-BPE}} (3);
				
				\draw [->, thick] (3) to node[right]  {\Cref{p:pd-BPE-search-to-pseudo_constr}} (4);
				
				\draw [->, bend right=10, thick] (4) to node[left]  {\Cref{t:hyp_vs_hier}} (5);
				
				\draw [<-, bend left=10, thick] (4) to node[right]  {\Cref{t:hierarchy-to-pseudo_constr}} (5);
				
				\draw [->, thick] (1) to node[left] 
				{{\cite[Theorem 3.6]{DBLP:conf/random/Barak02}}}
				(6);
				
				\draw [->, thick] (2) to node[above left]  {Trivial} (7);
				
				\draw [->, thick] (7) to node[right]  {\Cref{p:io-pd-to-hierarchy}} (6);
			\end{pgfonlayer}
		\end{tikzpicture}
	
	~\\
	
	An interesting question left open by our paper is to establish a converse to Proposition \ref{p:pd-BPE-search-to-pseudo_constr}.

\end{document}